\documentclass[journal,10pt]{IEEEtran}
\usepackage{graphicx}
\usepackage{comment}
\usepackage{amssymb}
\usepackage [noadjust]{cite} 
\usepackage{bm}  
\usepackage{amsmath}
\usepackage{amssymb}
\usepackage{enumerate}
\usepackage{stfloats}
\usepackage{cases}
\usepackage{epstopdf}
\usepackage{soul,color} 
\usepackage[percent]{overpic}
\usepackage{tabularx}
\usepackage{subfigure}
\usepackage[table]{xcolor}
\usepackage{multirow}
\usepackage{booktabs}  
\usepackage{tabu} 
\usepackage{soul}
\usepackage{algorithm}
\usepackage{algorithmic}
\usepackage{amsthm}

\newtheorem{lemma}{Lemma}
\newtheorem{corollary}{Corollary}

 \newcommand{\algstartblock}
 {	
 	\ifnum\value{algcounter}<10{
 		\addtocounter{algblockspace}{12}
 	}\else{
 		\addtocounter{algblockspace}{12}
 	}
 	\fi
 }
 \newcommand{\algendblock}{
 	\ifnum\value{algcounter}<10{
 		\addtocounter{algblockspace}{-12}
 	}\else{
 		\addtocounter{algblockspace}{-12}
 	}
 	\fi
 }

 \newcommand{\alg}[1]{
	\ifnum\value{algcounter}<10{
	 	\begin{enumerate}[\arabic{algcounter}.\hspace{\numexpr\value{algblockspace}+0 pt}]
	 		\setlength\itemindent{-7pt}
			\item \hspace{-8pt} #1
		\end{enumerate} \stepcounter{algcounter}
	}\else{
		\begin{enumerate}[\arabic{algcounter}.\hspace{\numexpr\value{algblockspace}+0 pt}]
			\setlength\itemindent{-11pt}
			\item \hspace{-8pt} #1
		\end{enumerate} \stepcounter{algcounter}
	}
	\fi
  }

\usepackage[utf8]{inputenc}
\usepackage[english]{babel}

\usepackage{amsthm}

 \newcommand{\Rdddirectone}[3]{R^{dir_1}}
 \newcommand{\Rdddirectzero}[2]{R^{dir_0}}

 \newcommand{\Rddrelay}[2]{R^{rel}}

 \newcommand{\Rcellularzero}[2]{R^{cel_0}}
 \newcommand{\Rcellularone}[3]{R^{cel_1}}

\begin{document}
\raggedbottom
\allowdisplaybreaks
\title{Mobility-Aware Performance in Hybrid  RF and Terahertz Wireless Networks}
\author{Md Tanvir Hossan, {\em  Member IEEE} and~Hina~Tabassum, {\em Senior Member IEEE}
     \thanks{Md Tanvir Hossan and H. Tabassum are with the Department of Electrical Engineering and Computer Science,  York University, (e-mail: mthossan@ieee.org, hinat@yorku.ca). This work is supported by the Discovery Grant from the Natural Sciences and Engineering Research Council of Canada.}}
\date{}
\maketitle

\begin{abstract}

Using tools from stochastic geometry, this paper develops a tractable framework to analyze the performance of a mobile user in a two-tier wireless  network operating on sub-6GHz  and terahertz (THz) transmission frequencies. Specifically,  using  an equivalence distance approach, we characterize the overall handoff (HO) probability  in terms of the horizontal and vertical HO and mobility-aware coverage probability.  In addition, we characterize novel coverage probability expressions for THz network in the presence of molecular absorption noise and highlight its significant impact on the users' performance. Specifically, we derive a novel closed-form expression for the Laplace Transform of the cumulative molecular noise and interference observed by a mobile user in a hybrid RF-THz network. Furthermore, we provide a novel approximation to derive the conditional distance distributions of a typical user in a hybrid RF-THz network. Finally, using the overall HO probability  and coverage probability expressions, the  mobility-aware probability of coverage has been derived  in a  hybrid RF-THz network. Our mathematical results validate the correctness of the derived expressions using Monte-Carlo simulations. The results offer insights into the adverse impact of users' mobility and molecular noise in THz transmissions on the probability of coverage of  mobile users. Our results demonstrate that a small increase in the intensity of terahertz base-stations (TBSs) (about 5 times) can increase the HO probability much more compared to the case when the intensity of RF BSs (RBSs) is increased by 100 times. Furthermore, we note that high molecular absorption can be beneficial (in terms of minimizing interference and molecular noise)  for specific deployment intensity  of TBSs and the benefits can outweigh the drawbacks of signal degradation due to molecular absorption.

\end{abstract}

\raggedbottom

\section{Introduction}
Connected and autonomous vehicles (CAVs) are becoming crucial nowadays to  improve the driving safety, ameliorate travel efficiency through efficient parking and routing, and minimize traffic congestion. In this context, ultra-reliable  and low latency communication (URLLC) is necessary to enable the exchange of  real-time information between vehicles, and vehicle to infrastructure; thereby enabling vehicles (or drivers) to make informed decisions. However, unfortunately, while the conventional sub-6GHz network benefits from strong transmission powers and wider coverage zones, it may not guarantee URLLC due to extremely limited and congested spectrum. In the sequel, transmissions at millimeter-wave (mmWave) ($\sim$ 30 - 100GHz) and terahertz (THz) ($\sim$ 0.1 - 10 THz) frequencies  will complement traditional wireless transmissions at sub-6GHz (or radio frequency (RF))  to support ubiquitous vehicular communications. 

To date, the THz spectrum which lies in between the mmWave and the optical spectrum has been investigated rarely. However, with  the recent innovations in THz signal generation,  radiation, and modulation methods,  the so-called THz gap is closing. THz spectrum can support  massive data rates in the order of hundreds of  Gigabits-per-second (Gbps), massive connectivity,  and extremely secure transmissions. Nevertheless, THz channel propagation is susceptible to unique challenges such as
\textcolor{black}{{molecular absorption noise\footnote{A part of electromagnetic energy gets transformed into the internal energy of molecules, referred to as molecular absorption noise which is a function of frequency.}, {varying molecular absorption coefficients} at different frequencies, and a sophisticated {Beer's Lambert law-based channel propagation}}} model. 


\textcolor{black}{While THz transmissions suffer from unfavorable propagation and atmospheric absorption; there are several reasons to explore THz bands for mobility-based applications as also noted in \cite{7982949,url}, i.e., \textbf{(i)} even if users are mobile, very high data rate transmission links become nearly static from the data viewpoint, i.e., the transmissions become almost “instantaneous.” In other words, although users' channel characteristics can vary over time, the variations happen at a much slower rate than the actual data rate transmission  \cite{7982949}, \textbf{(ii)} Even with the intermittent connectivity of a mobile user (e.g., a vehicle connecting to nearby access points), the amount of data that can be transmitted per connection is huge (i.e., 1 trillion bits in 1 second) \cite{7982949,url}. Thus, with faster communication, it is not necessary to be connected all the time. As long as the high-speed connection is available every now and then, the users can transfer or request all the data \cite{url}, \textbf{(iii)} by moving to higher carrier frequencies, the impact
of Doppler effect can be minimized which is crucial for 
transmissions to trains/aircrafts moving at high
speeds.}

\subsection{Background Work}
To date, a variety of research papers analyzed the coverage performance  considering a stand-alone THz network \cite{Joonaskokkoniemi2017stochastic,8763780, 8849958,9148716}.  In  \cite{Joonaskokkoniemi2017stochastic}, 
the authors characterized the average interference in a stand-alone THz network by applying the methods from the stochastic geometry assuming an interference-limited regime. However, the average interference expression was not applied to the coverage or outage analysis of a typical user. Instead, the interference distribution was approximated with a log-logistic distribution to compute the coverage probability. The authors highlighted that the application of the log-logistic approximation may not always be {precise}. In \cite{8763780}, the authors considered a stand-alone THz network to calculate the end-to-end latency and reliability, while assuming a Gaussian distribution of the interference. Likewise, in \cite{8849958}, the interference was approximated with the average interference.  

The aforementioned research works {examined} the stand-alone THz networks performance.  Recently, a mixed THz and RF decode-and-forward relaying system was studied  in \cite{8932597}. The authors derived the cumulative density function (CDF) of the receiver's end-to-end (E2E) signal-to-noise ratio (SNR), the outage probability, and symbol error rate (SER).
The authors in \cite{ntontin2016toward}  derived the approximate coverage probability  for a single-tier network, where {RBS or TBS can be used in} an opportunistic manner. However, given the small coverage of THz transmissions, it is practical to consider a two-tier network with a separate deployment of TBSs which is likely much denser than the deployment of RBSs. Different from the existing research,  \cite{9119462}  characterized the exact coverage probability and interference statistics of users in a {stand-alone} THz network and a hybrid two-tier RF-THz network {with the help of} stochastic geometry. 

\textcolor{black}{
None of the research works analyzed the impact of mobility on the performance of THz networks or multi-band networks. Also, the  impact of molecular absorption noise  was not considered in the stochastic-geometry based coverage analysis.}

\textcolor{black}{To date, several interesting research works have considered the impact of mobility in RF \cite{7006787, 7866856, 9023398}, or mm-wave networks \cite{mm1,mm2,mm3,mm4,1702.02775}.   In \cite{7006787}, the HO probability analysis was conducted in a multi-tier cellular network. The authors showed that there is an impact of users' mobility on tier association and coverage probability. Nevertheless, the framework in \cite{7006787} only deals with the horizontal HO (i.e., the HO between the BSs in the same tier). This shortcoming arises because the closest BS to the user after HO is always considered as the new serving BS, which is not true in multi-tier networks with BSs having distinct powers, coverage zones, and operating frequencies. To overcome this shortcoming, in \cite{7866856, 9023398}, the authors applied an equivalence-based approach to analyze both the vertical and horizontal HO probabilities in a two-tier RF network.  
Nevertheless, to attain high speed connectivity and URLLC in 6G, it is imperative to understand the impact of mobility in multi-band wireless networks. Interestingly, in \cite{1702.02775}, the authors introduce a software defined network (SDN) switching framework for vehicles equipped with transceivers capable of dynamically switching between THz and mmWave bands to accommodate asymmetric uplink/downlink communication.  }

\subsection{Motivation and Contributions}
\textcolor{black}{None of the research works presented a systematic stochastic geometry framework where the \textit{equivalent distance approach} has been applied to characterize the horizontal handoff  and vertical handoff  probabilities (or rate)\footnote{There are two  definitions of the \textit{HO rate} that exist in the literature. In \cite{6477064}, the HO rate is  the ratio of the average number of cells a mobile user traverses to the average transition time (including the pause time). In \cite{7006787}, the HO rate is defined as the probability that the user crosses over to the next cell in one movement period.  In this paper, we use the definition in  \cite{7006787}.}, and mobility-aware coverage probability  of a user in a two-tier \textit{multi-band network operating on different frequencies}. 
Furthermore, characterizing  the aforementioned performance metrics in the presence of THz transmissions brings additional novelty due to the unique features of THz that are \textit{different from conventional RF and mm-wave}, such as \textbf{(i}) \textit{molecular absorption noise} in the SINR expression,  and a \textbf{(ii})  sophisticated \textit{Beer's Lambert law-based channel propagation} model. }
To this end, the contributions of this paper are:



\begin{itemize}
    \item  We characterize  the overall HO probability  (which is based on the vertical and horizontal HO probability) of a mobile user in the downlink of a hybrid RF-THz network, considering the maximum received signal power association criterion. In this context, 
    we apply  an \textit{equivalence distance approach} to facilitate the analysis of vertical HO, i.e., by introducing a virtual tier with the serving tier of a mobile user.  In addition,  we pointed out that a correction factor is missing in the HO probability expressions of all aforementioned research works (whether single-tier or multi-tier) related to mobility {\cite{7006787, 7866856, 9023398}}. 
    \item We analyze the exact coverage probability of a typical user in the THz networks considering the repercussions of molecular noise absorption and highlight the devastating impact of ignoring the molecular absorption noise on the coverage probability. Specifically, we derive a new closed-form expression for the Laplace Transform (LT) of the cumulative molecular noise and interference observed by a typical user in THz network.
    \item We provide a novel approximation to derive the conditional distance distribution of a typical user in a hybrid  network. To tackle mathematically challenging  Beer's Lambert transmission model,  we propose a novel and efficient approximation to make the framework tractable.
    \item Using the overall HO probability  and coverage probability, we derive the  \textit{mobility-aware probability of coverage} of a mobile user in a  hybrid RF-THz network. 
    \item \textcolor{black}{We provide an  asymptotic closed-form expression of association probability for low molecular absorption coefficient, and asymptotic single-integral  expression of no HO probability when the users move in a straight line.}
   \item  Numerical results validate  the accuracy of our derived expressions.  \textcolor{black}{The derived expressions can be computed numerically using standard mathematical software such as \texttt{MAPLE} and \texttt{Mathematica} to obtain useful insights related to the user's performance in a hybrid RF/THz network with mobility and molecular absorption noise.}
\end{itemize}

\begin{table*}[t]
\centering
\tiny
\caption{{\color{black} List of Notations} }
\resizebox{\textwidth}{!}{\begin{tabular}{|c|c|c|c|}
\hline
  \textbf{Symbol}&\textbf{Description} & \textbf{Symbol}&\textbf{Description}\\ 
  \hline
  \footnotesize $\mathbf{\Phi}_{R}$ & \footnotesize Locations  of  the  conventional RBSs  & \footnotesize  $R_{\mathrm{th}}$ & \footnotesize Desired target rate \\ 
  \hline
  \footnotesize $\mathbf{\Phi}_{T}$ & \footnotesize Locations  of  the  conventional TBSs  & \footnotesize $\theta$ & \footnotesize Boresight direction angle \\ 
  \hline
  \footnotesize $\mathbf{\Phi}_{u}$ & \footnotesize  Locations  of  the  user & \footnotesize $w_{q}$ & \footnotesize Main lobe beam-width  \\ 
  \hline
  \footnotesize $\lambda_{R}$ & \footnotesize Intensity of RBS  & \footnotesize $G^q_{\mathrm{max}}$ & \footnotesize Beamforming gains of main lobes  \\ 
  \hline
  \footnotesize $\lambda_{T}$ & \footnotesize Intensity of TBS  & \footnotesize $G^q_{\mathrm{min}}$ & \footnotesize Beamforming  gains of side lobes \\ 
  \hline
  \footnotesize $\lambda_{u}$ & \footnotesize Intensity of user  & \footnotesize $N_T$ & \footnotesize Noise originates from thermal and molecular absorption \\ 
  \hline
  \footnotesize $P_{R}^{\mathrm{tx}}$ & \footnotesize Transmit power  of  the  RBSs  & \footnotesize $N_{0}$ & \footnotesize Thermal noise \\ 
  \hline
  \footnotesize $G_{R}^{\mathrm{tx}}$ & \footnotesize Transmitting  antenna  gain RF & \footnotesize $r_{0}$ & \footnotesize Distance between the mobile user to the serving  RBS \\ 
  \hline
  \footnotesize $G_{R}^{\mathrm{rx}}$ & \footnotesize Receiving antenna gain of RF  & \footnotesize  $r_{i}$ & \footnotesize Distance between  the \textit{i}-th  interfering  RBS  and user \\ 
  \hline
  \footnotesize $c$ & \footnotesize Speed of the electromagnetic  wave  & \footnotesize $r_T$ & \footnotesize Distance  between  initial location of  the user  and  TBS \\ 
  \hline
  \footnotesize $f_{R}$ & \footnotesize RF  carrier frequency (in GHz)  & \footnotesize $R_T$ & \footnotesize Distance  between  final location of  the user  and  TBS \\ 
  \hline
  \footnotesize $\alpha$ & \footnotesize Path-loss exponent for the RF signal  & \footnotesize $r_R$ & \footnotesize Distance between initial location of the user and RBS  \\ 
  \hline
  \footnotesize $H$ & \footnotesize Fading channel power  & \footnotesize $R_R$ & \footnotesize Distance between final location of the user and RBS \\ 
  \hline
  \footnotesize $N_R$ & \footnotesize Thermal noise   & \footnotesize  $r_{T}^{\prime}$ & \footnotesize Equivalent distance of $r_{R}$ \\ 
  \hline
  \footnotesize $I_{R}$ & \footnotesize cumulative interference  from interfering RBSs & \footnotesize  $R_{T}^{\prime}$ & \footnotesize Equivalent distance of $R_{R}$ \\ 
  \hline
  \footnotesize $H_{i}$ & \footnotesize Fading from the \textit{i}-th interfering RBS   & \footnotesize $r_{R}^{\prime}$ & \footnotesize Equivalent  distance  of $r_T$ \\ 
  \hline
  \footnotesize $P_{T}^{\mathrm{tx}}$ & \footnotesize  Transmit  power  of  the  TBSs & \footnotesize $R_{R}^{\prime}$ & \footnotesize Equivalent  distance  of $R_T$ \\ 
  \hline
  \footnotesize $G_{T}^{\mathrm{tx}}$ & \footnotesize Transmitting  antenna  gain  of  TBSs  & \footnotesize $A_R$ & \footnotesize Association probability of a user with RBS \\ 
  \hline
  \footnotesize $G_{T}^{\mathrm{rx}}$ & \footnotesize Receiving  antenna  gain of  TBSs & \footnotesize $A_T$ & \footnotesize Association probability of a user with TBS \\ 
  \hline
  \footnotesize $f_{T}$ & \footnotesize Carrier  frequency  in  THz  & \footnotesize $\mathbb{P}(H)$ & \footnotesize Overall HO probability of the typical user \\ 
  \hline
  \footnotesize  $K_{a}(f_T)$ & \footnotesize Molecular absorption coefficient  & \footnotesize $v$ & \footnotesize Typical user velocity\\
  \hline
  \hline
\end{tabular}}
\label{Notation_Summary_mmwave}
\end{table*}
\normalsize

The rest of the paper is organized as follows.  Section II presents the system model, assumptions, and  the methodology of analysis. The horizontal and vertical HO probability analysis is presented  in Section III. Section IV characterizes the coverage probability of a user in the presence of molecular absorption noise and incorporates  the impact of horizontal and vertical HO probability in the calculation of the mobility-aware coverage probability. Finally, selected numerical and simulation results are presented in Section V before  conclusions in Section~VI.

\section{Network Model and Assumptions}
In this section, we present  the spatial network deployment  model, channel propagation models, and mobility/HO model of a typical mobile user in a  {multi-band}  network. Finally, we present the step-by-step methodology of analyzing the mobility-aware coverage probability.

\subsection{Spatial Network Deployment}
A two-tier downlink network comprised of a layer of RF BSs (RBSs) and a layer of THz BSs (TBSs) is considered. The spatial deployment of the  RBSs and TBSs is taken as a two-dimensional (2D) homogeneous Poisson Point Processes (PPP)  $\mathbf{\Phi}_{R}$ and $\mathbf{\Phi}_{T}$ with intensities  $\lambda_{R}$ and $\lambda_{T}$, respectively.  \textcolor{black}{We evaluate the performance of a mobile user  who is originally found at the origin and measures  the channel quality from RBSs and TBSs as is done in the existing heterogeneous networks (HetNets). Users then handoff opportunistically  to the RBS or TBS  and the BS serves various users in orthogonal channels or time slots. The mobile user connects to a given BS based on maximum received signal power.} An illustration of the considered network is shown in Fig.~1 where the typical mobile user can be classified according to its velocity, e.g., pedestrians with low velocity and vehicles with moderate to high velocity.

\begin{figure*}\label{sys_mod}
\centering
\includegraphics[width=6.5in]{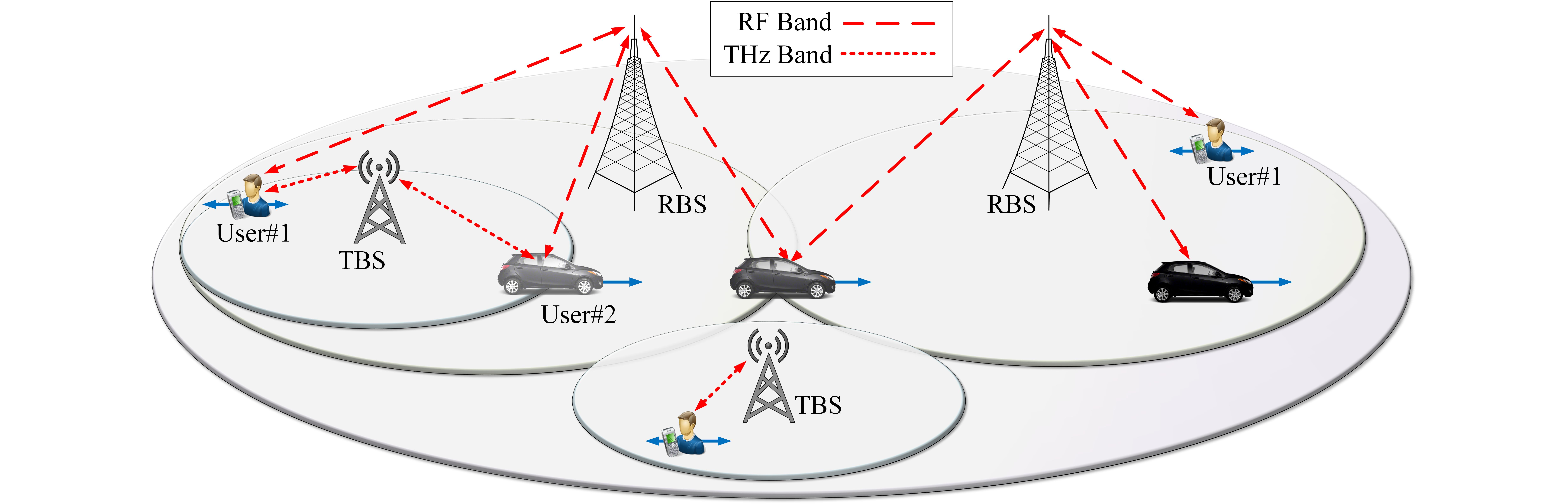}
\caption{{Graphical illustration of a two-tier hybrid RF-THz network with low- and high-velocity users.}}
\label{u3}
\end{figure*}
\subsection{Mobility Model and HO Criterion}
The typical mobile user moves with a velocity $v$ from the origin in an arbitrary direction, thereby HO may occur depending on the maximum received signal power criterion.  HOs (or association of users) can be performed based on both  the instantaneous
received power \cite{dhillon2011tractable,6171996} and maximum long-term averaged received power  \cite{jo2012heterogeneous, ye2013user,rubio2017user,ye2016user,zappone2018user}.  However, the short-term instantaneous fading can yield  unnecessary HOs, that is, the ``ping-pong effect``. To overcome this undesired phenomenon, the received signal power is averaged over the measurement period  in long-term evolution (LTE). This assumption, also has been considered in other research works \cite{ye2013user,rubio2017user,ye2016user,zappone2018user} and is considered as more realistic compared to instantaneous received power based user association \cite[page1]{zappone2018user}. 
The HOs in the same tier (e.g., RBS-RBS or TBS-TBS) are referred to as \textit{horizontal HO}. Alternatively, when the type of user switches its BSs in two different tiers (e.g., RBS-TBS or TBS-RBS), then this HO is referred to as \textit{vertical HO}.

\subsection{RF and THz Communication Model}
\subsubsection{RF Model}
The  signal transmitted from RBS incurs path-loss and short-term fading which is Rayleigh distributed. At the typical mobile user, the received signal power is defined as:
\begin{equation}\label{RF_dis}
  P_{R}^{\mathrm{rx}} = G_{R}^{\mathrm{tx}}\:G_{R}^{\mathrm{rx}}\:\left( \frac{c}{4\pi f_{R}} \right)^{2} \frac{P_{R}^{\mathrm{tx}}}{r_{0}^{\alpha}},
\end{equation}
The signal-to-interference-plus noise ratio (SINR) of a typical mobile user on RF transmission channel is thus modeled as:
\begin{equation}{\label{sinr_rf}}
    \mathrm{SINR}_{\mathrm{R}} = \frac{P_{R}^{\mathrm{tx}}\:G_{R}^{\mathrm{tx}}\:G_{R}^{\mathrm{rx}} \left(\frac{c}{4\pi f_{R}} \right)^2 H }{r_{0}^{\alpha}\left(N_{R} + I_{R}\right)}=  \frac{P_{R}^{\mathrm{tx}} \gamma_R H }{r_{0}^{\alpha}\left(N_{R} + I_{R}\right)},
\end{equation}
where  $P_{R}^{\mathrm{tx}}, G_{R}^{\mathrm{tx}}, G_{R}^{\mathrm{rx}}, c, f_{R}, r_{0},$ and $\alpha$ denote the transmit power from the RBSs, transmitting antenna gain, receiving antenna gain, speed of the electromagnetic wave, RF carrier frequency (in GHz), distance between the mobile user to the serving RBS, and path-loss exponent of the signal, respectively. 
Also, $H$ is the exponentially distributed channel fading power of the mobile user from the targeted RBS,  $N_R$ is the power of thermal noise  at the receiver, $I_{R} = \sum_{i\in \Phi_{R}\backslash 0} P_{R}^{\mathrm{tx}} \gamma_{R} r_{i}^{-\alpha} H_{i}$  is the cumulative interference at the mobile user from the interfering RBSs. From the cumulative interference, $r_{i}$ is the distance between the $i$-th interfering RBS and the typical mobile user, $H_{i}$ is the power of fading from the $i$-th interfering RBS to the typical mobile user, and $ \gamma_R = G_{R}^{\mathrm{tx}} \: G_{R}^{\mathrm{rx}} \left( {c}/{4\pi f_{R}} \right)^2 $. 

\subsubsection{THz Model}
In THz network, the  line-of-sight (LOS) transmissions are  much more significant than the non-line-of-sight (NLOS) transmissions due to the presence of molecular absorption. Subsequently, in this work, following  \cite{Joonaskokkoniemi2017stochastic,7982949,8763780}, we calculate the received power taking into account the LOS transmission property from \cite{7820226}, \cite{9119462} as follows\footnote{\textcolor{black}{
THz transmissions  are   prone to the the molecular absorption in the (indoor/outdoor) atmosphere. This absorption process can be described with the help of \textit{Beer-Lambert’s law} which states that the amount of radiation  that is able to propagate from a transmitter to the receiver through the absorbing medium can be characterized  by
$\mathrm{exp}(-K_{a}(f_T)\:d_{0})$, where $K_a(f_T)$ denotes the molecular absorption coefficient of the indoor or outdoor medium \cite{ma2018invited, federici2016review}.  The model is shown to  be applicable to both indoor and outdoor scenarios \cite{ma2018invited, federici2016review}.   Our contributions in this paper are general and are applicable for any values of $K_a(f_T)$. 
}}:
\begin{equation}\label{THz_dis}
   P_{T}^{\mathrm{rx}} =  G_{T}^{\mathrm{tx}}\:G_{T}^{\mathrm{rx}} \left( \frac{c}{4\pi f_{T}} \right)^{2} \: \frac{P_{T}^{\mathrm{tx}}\:\mathrm{exp}(-K_{a}(f_T)\:d_{0})}{d_{0}^{2}},
\end{equation}
where $P_{T}^{\mathrm{tx}}, G_{T}^{\mathrm{tx}}, G_{T}^{\mathrm{rx}}, f_{T}, d_{0},$ and $K_{a}(f_T)$ denote the transmit power of the TBSs, transmitting antenna gain of the TBS, receiving antenna gain of the TBS, THz carrier frequency, distance between the mobile user to the serving TBS, and the molecular absorption coefficient depends on the composition of the medium and also on the frequency (i.e., $f_{T}$) of the signal, respectively. 
\textcolor{black}{
For any specific THz carrier frequency $f_T$,  $K_a(f_T)$\footnote{For the sake of brevity, we will drop the argument of $K_a(f_T)$ from this point onwards in the paper.}  can be calculated as follows \cite{5995306}:
\small
\begin{equation}\label{Ka(f)}
{K_a(f_T)}\mathrm = \sum_{(i,g)} {\frac{p^2 T_{ \mathrm {sp}}  q^{(i,g)}  {N_A} S^{(i,g)} f \tanh{\left( \frac{h  c  f}{2 k_b T}\right)}} {p_0  V T^2 f^{(i,g)}_c \tanh{\left(\frac{h  c f^{(i,g)}_c}{2 k_b T}\right)}}}   F^{(i,g)}\left(f\right),
\end{equation}
\normalsize
where \(p\) and \(p_0\) are the ambient pressure of the transmission medium and the reference pressure, respectively, \(T\) is the temperature of the transmission medium, \(T_{\mathrm{sp}}\) is the temperature at standard pressure, \(q^{(i,g)}\) indicates the mixing ratio of the isotopologue \(i\) of gas \(g\), \(N_A\) refers to the Avogadro number, and \(V\) is the gas constant. The line intensity \(S^{(i,g)}\)  defines the strength of the absorption by a specific type of molecules and is directly obtained from the HITRAN database \cite{ rothman2009hitran}. In addition, \(f\) and \(f^{(i,g)}_c\) denote the THz frequency  and the resonant frequency of gas \(g\), respectively, \(c\) is the speed of light, \(h\) is the Planck's constant, and \(k_b\) is  the Boltzmann constant. For the frequency $f$, we consider the Van Vleck-Weisskopf asymmetric line shape  to evaluate:
\small
\begin{equation}
F^{(i,g)}(f)= \frac{100 \:c \:\alpha^{(i,g)} f} {\pi \:f_c^{(i,g)}} \left(\frac{1}{Y^2+ (\alpha^{(i,g)})^2}+\frac{1}{Z^2+(\alpha^{(i,g)})^2}\right),\nonumber
\end{equation}
\normalsize
where $Y= f+f_c^{(i,g)}$ and $Z= f-f_c^{(i,g)}$, and the Lorentz half-width is given as follows:
$$\alpha^{(i,g)}= \left( \left(1 - q^{(i,g)} \right) \alpha_{\mathrm {air}}^{(i,g)}+ q^{(i,g)} \alpha_0^{(i,g)}\right) \left(\frac{p}{p_0}\right) \left(\frac{T_0}{T}\right)^{\gamma},$$ where \(T_0\) is the reference temperature, the parameters  air half-widths, \(\alpha_{\mathrm{air}}^{(i,g)}\), self-broadened half-widths, \(\alpha_0^{(i,g)}\), and temperature broadening coefficient, \(\gamma\), are obtained  from the HITRAN database \cite{rothman2009hitran}. The resonant frequency of gas $g$ at reference pressure \(p_0\) is determined as \( f_c^{(i,g)}=f_{ {c_0}}^{(i,g)} + {\delta}^{{(i,g)}}{(\frac{p}{p_0})} \), where \(\delta^{(i,g)}\) is the linear pressure shift \cite{5995306}.
}

Note that $G_{T}^{\mathrm{tx}}\left(\theta_q\right)$ as well as $G_{T}^{\mathrm{rx}}\left(\theta_q\right)$ are directional transmitter and receiver antenna gains, respectively. 
The  beamforming gains from the main lobe and side lobes of the TBS transmitting antenna can be generalized as follows \cite{di2015stochastic}:   
\textcolor{black}{
\begin{equation}
\label{eq:gain2}
  G_{T}^{\mathrm{q}}\left(\theta\right) =
    \begin{cases}
      G^q_{\mathrm{max}} & \mid \theta_q \mid \leq w_{q}\\ 
      G^q_{\mathrm{min}} & \mid \theta_q \mid > w_{q}
    \end{cases},  
\end{equation}
where $q\in \{\mathrm{tx,rx}\}$, $\theta_q \in [-\pi,\pi)$
     is the angle off the boresight direction, $w_{q}$ is the beamwidth of the main lobe, $G^q_{\mathrm{max}}$ and $G^q_{\mathrm{min}}$ are the beamforming gains of the main and side lobes, respectively.} We assume that the typical mobile user's receiving beam aligns with the transmitting beam of the associated TBS through beam alignment techniques. However, for the alignment between the user and interfering TBSs, we define a random variable $D$, which can take values as  $D \in \{G^{\mathrm{tx}}_{\mathrm{max}}G^{\mathrm{rx}}_{\mathrm{max}},G^{\mathrm{tx}}_{\mathrm{max}}G^{\mathrm{rx}}_{\mathrm{min}},G^{\mathrm{tx}}_{\mathrm{min}}G^{\mathrm{rx}}_{\mathrm{max}},G^{\mathrm{tx}}_{\mathrm{min}}G^{\mathrm{rx}}_{\mathrm{min}}\},$ and the respective probability for each case is $F_{\mathrm{tx}}F_{\mathrm{rx}}$, $F_{\mathrm{tx}}(1-F_{\mathrm{rx}})$, $(1-F_{\mathrm{tx}})F_{\mathrm{rx}}$, and $(1-F_{\mathrm{tx}})(1-F_{\mathrm{rx}})$, where $F_{\mathrm{tx}} = \frac{\theta_{\mathrm{tx}}}{2\pi}$ and $F_{\mathrm{rx}} = \frac{\theta_{\mathrm{rx}}}{2\pi}$, respectively. {Assuming that the main lobe of the typical mobile user's receiver  is coinciding with that of its desired TBS\footnote{\textcolor{black}{The  beamforming model is based on a two-lobe approximation of the antenna pattern. Although simple, the model is tractable and capture primary features such as the directivity gain, the front-to-back ratio, and the half-power beamwidth~\cite{di2015stochastic,6840343}.} }, its SINR the can be formulated as follows:}
\begin{align}{\label{sinr_1}}
    &\mathrm{SINR}_{T} = \frac{G_{T}^{\mathrm{tx}} G_{T}^{\mathrm{rx}} \left(\frac{c}{4\pi f_{T}} \right)^2 P_{T}^{\mathrm{tx}}\: \mathrm{exp}(-K_a d_{0})d_{0}^{-2} }{N_{T} + I_{T}},
\nonumber\\&=
    \frac{ \gamma_T \: P_{T}^{\mathrm{tx}}\:\mathrm{exp}(-K_a d_{0})d_{0}^{-2}}{N_T+I_T},
\end{align}
where $I_{T} = \sum_{i\in \Phi_T \setminus 0}\gamma_T \:P_{T}^{\mathrm{tx}}\:F {d_{i}}^{-2} \mathrm{exp}(-K_a \:{d_{i}})$ is the cumulative interference at the mobile user, $d_{i}$ is the distance of the that user to the $i$-th interfering TBS, $F = F_{\mathrm{tx}}F_{\mathrm{rx}}= \frac{\theta_{\mathrm{tx}} \theta_{\mathrm{rx}}}{4 \pi^2}$ is the probability of alignment between the main lobes of the interferer and the typical user assuming negligible side-lobe gains
and  $\gamma_{T} = G^{\mathrm{tx}}_{T}\:G^{\mathrm{rx}}_{T}\:{c^2}/{\left(4\pi f_{T}\right)^2}$.  The cumulative thermal and molecular absorption noise is \cite{7390991}, \cite{8763780}, \cite{9322606}:
\begin{align}
    N_T= N_{0} +  \:P_{T}^{\mathrm{tx}} \gamma_T \:{d_{0}^{-2}} \:(1-e^{-K_a \:{d_{0}}})+\nonumber\\ \sum_{i\in \Phi_T \backslash 0} \gamma_T F \:P_{T}^{\mathrm{tx}} \:{d_{i}^{-2}}(1-\mathrm{exp}(-K_a\:{d_{i}})).
\end{align}
\textcolor{black}{Note that the internal vibration of the molecules re-emit a part of the absorbed energy back to  the  channel resulting in the  so-called  \textit{molecular absorption  noise} \cite{5995306,7248500, Kokkoniemi2016}. The molecular absorption noise is induced by the transmissions of the users sharing the same frequency. As such, the second and third terms in $N_T$ represent the molecular absorption noise due to the desired users' transmission and the interfering users' transmission, respectively.} 
The SINR from TBS can then  be modeled  as follows:
\small
     \begin{align}
    &\mathrm{SINR}_{T} = \frac{P_{T}^{\mathrm{tx}} \gamma_T  d_{0}^{-2} {e^{-K_{a}d_{0}}}}{ N_{0} + \:P_{T}^{\mathrm{tx}} \gamma_T \:{d_{0}^{-2}} \:(1-e^{-K_a \:{d_{0}}})+ \sum_{i\in \Phi_T \backslash 0} \:P_{T}^{\mathrm{tx}} \gamma_T F \:{d_{i}^{-2}}}. \nonumber
\end{align}
\normalsize
\textcolor{black}{The SINR expression is different from the traditional RF  systems due to THz channel propagation model in the numerator and molecular noise consideration in the denominator.}

\subsection{Methodology of Analysis}
The methodology of analyzing the  HO probability  and mobility-aware coverage probability in a {multi-band} network can be summarized as follows:
\begin{itemize}
\item \textcolor{black}{Using Eq.(3),} derive the conditional probability density function (PDF) of the distance of a {mobile} user  {tagged} to the TBS ($f_{r_T}(r_T)$) and RBS ($f_{r_R}(r_R)$) in a {multi-band} network.
\item \textcolor{black}{Using Eq.(3),} derive the conditional HO probability of a typical user who is initially associated to TBS ($\mathbb{P}(H_T)$) and  initially associated to RBS ($\mathbb{P}(H_R)$).
\item \textcolor{black}{Using Eq.(3)} and  the association probabilities of the typical user to TBSs and RBSs, i.e., $A_T$ and $A_R$, respectively, and conditional HO probabilities $\mathbb{P}(H_R)$ and $\mathbb{P}(H_T)$, derive the overall HO probability, i.e., $\mathbb{P}(H)$ of the typical user.
\item \textcolor{black}{Using Eq.(3),} we derive the LT of the  cumulative interference and molecular noise as well as the coverage probability of the typical user without mobility $\mathbb{C}$.
\item Derive the coverage probability of the typical user  with mobility  $\mathbb{C}_M$.
\end{itemize}

\section{HO Probability in a Hybrid RF-THz Network}
In this section, first develop HO criterion from TBS and derive the conditional HO probability from TBS, i.e., $\mathbb{P}(H_T)$, which comprises of the HO probability from TBS to TBS (horizontal HO)  and TBS to RBS (vertical HO). Then, we formulate  and simplify  the HO criterion from RBS and derive the HO probability from RBS, i.e,  $\mathbb{P}(H_R)$, which comprises of the HO probability from RBS to RBS (horizontal HO) and RBS to TBS (vertical HO). Finally, develop the overall HO probability of the mobile user, which is defined as:
\begin{equation}
    \mathbb{P}(H)=  A_R\mathbb{P}(H_R)+ A_T\mathbb{P}(H_T),
\end{equation}
where $A_R$ and $A_T$ denote the association probabilities of a user with RBS and TBS, respectively.

{From the relationship between the received powers of TBSs and RBSs, the association probability to TBS can be defined as follows:}
\begin{equation}
\label{eq:Asso_Probab_Proof}
\begin{split}
A_T & = \mathbb{E}_{d_0}\left[\mathbb{P}\left( P_{T}^{\mathrm{rx}}>P_{R}^{\mathrm{rx}}\right)\right], \\ 
& = \mathbb{E}_{d_0}\left[\mathbb{P}\left(P^{\mathrm{tx}}_{{T}}\gamma_{T}\frac{\exp\left(-K_{a} d_0\right)}{d_0^2}>P^{\mathrm{tx}}_{{R}}\gamma_{R} r_0^{-\alpha}\right)\right],\\ 
& \stackrel{(a)}{=} \mathbb{E}_{d_0}\left[\exp \left(-\pi\lambda_{R}\left(Q d_0^2\exp\left(K_{a} d_0\right)\right)^{\frac{2}{\alpha}}\right)\right],
\end{split}
\end{equation}
where from the null probability of PPP $\Phi_R$ and $Q = \frac{P_{R}^{\mathrm{tx}}}{P_{T}^{\mathrm{tx}}} \frac{\gamma_R}{\gamma_T}$, we get step $\left(a\right)$ in Eq. (\ref{eq:Asso_Probab_Proof}). 
This null property stated the probability that no RBSs are closer to a user than the distance $z$, which is $\mathbb{P}\left(\rho\geq z\right) = \exp \left(-\pi \lambda_{R}z^2\right)$ for given a tier of RBSs with intensity $\lambda_{\mathrm{R}}$. The PDFs of the distances  between the typical user and the closest RBS and TBS are given as $f_{r_0}(r_0) = 2\pi \lambda_{T}r_0\exp\left(-\pi\lambda_{T}r_0^2\right)$ and $f_{d_0}(d_0)= 2\pi \lambda_{T} d_0 \mathrm{exp}\left(-\pi \lambda_{T} d_0^{2}\right)$, respectively. Therefore, the association probability with TBS can be achieved by averaging over $d_0$. 
\begin{align}
    A_{T} =\int_{0}^{\infty} \mathrm{exp} \left(-\pi \lambda_R Q^{\frac{2}{\alpha}}\:d_0^{\frac{4}{\alpha}}\:\mathrm{exp}\left( \frac{2K_{a}\:d_0}{\alpha} \right)\right) f_{d_0}(d_0) dd_0,
    \label{assoc1}
\end{align}
where $A_R=1-A_T$. \textcolor{black}{When $\alpha=4$ and $K_a \rightarrow 0$, a closed-form expression can be given as:
\begin{align}
    &A_{T} =\int_{0}^{\infty} \mathrm{exp} \left(-\pi \lambda_R Q^{\frac{2}{\alpha}}\:d_0^{\frac{4}{\alpha}}\right) f_{d_0}(d_0) dd_0 
    \nonumber\\&= 1 - 0.5 \pi  \lambda_R \,e^{\frac{\pi Q \lambda_R^2}{4 \lambda_T}} \sqrt{\frac{Q}{\lambda_T}}\,\mathrm{erfc}\left(\frac{\lambda_R}{2}\sqrt{\frac{\pi Q}{\lambda_T}}\right),
\end{align}
}
\begin{figure*}[t!]\label{ho_scn}
\centering
\includegraphics[width=5in]{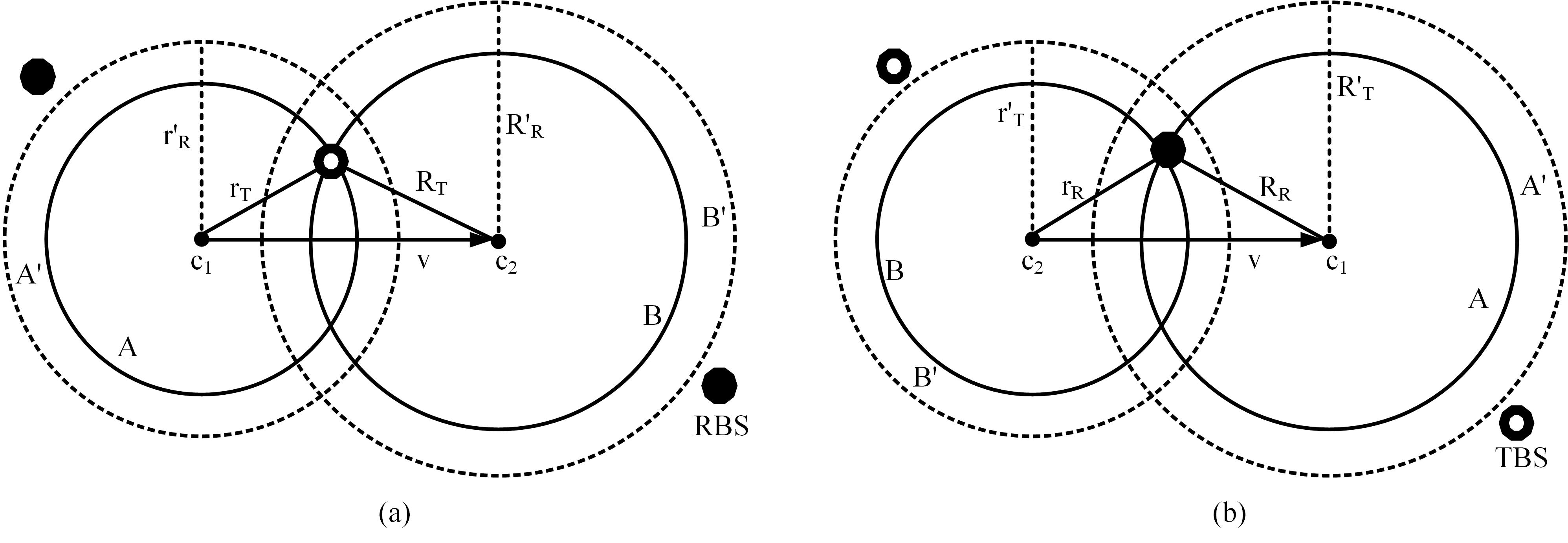}
\caption{Graphical illustration of different HO events, 
(a) HO from TBS: At $c_1$, a mobile user is initially associated with a TBS, where the distance between user and TBS is $r_T$. After HO at $c_2$,  the distance becomes $R_T$. The virtual tiers are shown by the dotted lines, i.e., the equivalent distance of the TBS from $c_1$ and $c_2$ in RF tier is represented by $r_{R}^{\prime}$ and $R_{R}^{\prime}$, respectively. 
(b) HO from RBS: At $c_2$, a mobile user is initially associated with a RBS, where the distance between user and RBS is $r_R$. After HO at $c_1$,  the distance becomes $R_R$. The virtual tiers are shown by the dotted lines, i.e., the equivalent distance of the RBS from $c_2$ and $c_1$ is represented by $r_{T}^{\prime}$ and $R_{T}^{\prime}$, respectively. 
}
\label{graph}
\end{figure*}

\subsection{HO Probability Characterization from TBS}
\subsubsection{HO Criterion from TBS}
Fig. \ref{graph}(a) illustrates an outline of an user, who is initially tagged with a TBS at the position $c_1$. Let $r_T$ is the distance between the user and the tagged TBS in a hybrid RF-THz network whose PDF is given as follows. 
\begin{lemma}
\label{lm:Distance_distribution}
The conditional PDF of the distance from a mobile user initially connected to TBS to the desired TBS is:
\begin{align}
\label{eq:Distance_distribution_THz}
&f_{r_{T}}(r_{T}) = \frac{2\pi \lambda_{T} \: r_{T}}{A_{T}} \times\nonumber\\& \mathrm{exp}\left(-\pi \lambda_{T} r_{T}^{2} -\pi \lambda_{R} (r_{T}^{2} Q)^{\frac{2}{\alpha}}\:\mathrm{exp}\left( \frac{2 K_{a} r_{T}}{\alpha} \right)\right).
\end{align}
\end{lemma}
\begin{proof}
See \textbf{Appendix~A}.
\end{proof}
The area $A$ is centered by $c_1$ with radius $r_T$. Assume that the  user moves its position from $c_1$ to $c_2$. The new distance $R_T$ denotes the distance between $c_2$ and the tagged TBS and  $B$ denotes the area centered at $c_2$ with radius $R_T$. At position $c_2$, \textit{the vertical HO takes place when the maximum received power of RBS is greater than the  TBS}, i.e.,  $P_{T}^{\mathrm{rx}} < P_{R}^{\mathrm{rx}}$, which results in:
\begin{align}\label{t2r_eqn}
     r_{R} < \mathrm{exp}\left( \frac{K_{a}\:r_{T}}{\alpha} \right)\left( Q\, r_{T}^2 \right)^{\frac{1}{\alpha}}\triangleq r_{R}^{\prime},
\end{align}
From Eq. (\ref{t2r_eqn}), we define $r_{R}^{\prime} $ is the equivalent distance of $r_{T}$. That is, when $r_{T}>r_{R}>r_{R}^{\prime}$,  vertical HO will not occur because it violates Eq. (\ref{t2r_eqn}). 
\subsubsection{HO Analysis}
Provided that the typical mobile user is {originally} tagged to TBS, the conditional HO probability from TBS can be determined by averaging over $r_T$ and $\theta$ as follows:
\begin{align}
    \mathbb{P}(H_T)=1-\mathbb{P}(\overline{H}_{T}) = 1-\mathbb{E}_{r_T,\theta}[\mathbb{P}(\overline{H}_{T}|r_T,\theta)].
\end{align}

To derive the no HO probability ($\mathbb{P}(\overline{H}_{T})$)  of a typical user who is associated to TBS, we replace the serving TBS by a virtual TBS in RF tier with distance $r_{R}^{\prime}$ away from the target mobile user. There will be no HO if  no RBSs or TBSs are closer to the user then  $r_{R}^{\prime}$.  In Fig. \ref{graph}(a), the area $A^{\prime}$ centered at $c_1$ with radius $r_{R}^{\prime}$. Here, the THz tier is the serving tier suggests all RBSs are found outside the area $A^{\prime}$. Likewise, $B^\prime$ is the area centered at $c_2$ with radius $R_{R}^{\prime}$, which is the corresponding distance of $R_{T}$. If RBSs are not remained within the area $B^{\prime}$, then initial TBS will remain the target BS even after the movement. 
\begin{lemma} Given a mobile user is initially associated with a TBS, the conditional probability of no HO from the serving TBS in a hybrid RF-THz network finds as follows:
\begin{align}
    &\mathbb{P}(\overline{H}_{T})  =  \frac{1}{\pi} \left(\int_{\theta = 0}^{\frac{\pi}{2}} \int_{r_{T} = 0}^{\infty} f_{r_{T}}(r_{T}) \: e^{ \left(-\lambda_{T} S_{T} - \lambda_{R} S_{T}^{\prime}\right)} d r_{T} d\theta \right.\nonumber\\&\left.
     + \int_{\theta = \frac{\pi}{2}}^{\pi} \int_{r_{T} = 0}^{v\:\mathrm{cos}(\pi - \theta)} f_{r_{T}}(r_{T}) e^{\left(-\lambda_{T} C_{T} - \lambda_{R} C_{T}^{\prime}\right)} d r_{T} d\theta 
   \right.\nonumber\\&\left. + \int_{\theta = \frac{\pi}{2}}^{\pi} \int_{r_{T} = v\:\mathrm{cos}(\pi - \theta)}^{\infty}  f_{r_{T}}(r_{T})\: e^{\left(-\lambda_{T} S_{T} - \lambda_{R} S_{T}^{\prime} \right)}  d r_{T} d\theta \right), \nonumber
\end{align}
where
$ S_{T} = |B| - |B\cap A|$,
$ S_{T}^{\prime} = |B^{\prime}| - |B^{\prime} \cap A^{\prime}|$,
$ C_{T} = R_{T}^{2}(\pi - \theta_{1}^{\prime T}) + r_{T} \: v \: \mathrm{sin}\theta - r_{T}^{2}(\pi - \theta),$
$ C_{T}^{\prime} = R_{R}^{\prime 2}(\pi - \theta_{3}^{\prime T}) + r_{R}^{\prime} \: v \: \mathrm{sin}\theta_{2}^{T} - r_{R}^{\prime 2}(\pi - \theta_{2}^{T}),$
$ \theta_{1}^{\prime T} = \theta - \pi + \mathrm{sin}^{-1} \left(\frac{v\:\mathrm{sin}\theta}{R_T} \right),$
$ \theta_{3}^{\prime T} = \theta - \pi + \mathrm{sin}^{-1} \left(\frac{v\:\mathrm{sin}\theta}{R_{R}^{\prime}} \right)$.
\end{lemma}
\begin{proof}
See \textbf{Appendix~B}.
\end{proof}
{In the following, as a special case of the aforementioned lemma, we have provided the HO probability of a mobile user when $\lambda_R \rightarrow 0$ which results in a stand-alone THz network.
\begin{corollary} As $\lambda_R \rightarrow 0$ then $A_T \rightarrow 1$, the conditional HO probability  can  be simplified as:
\begin{align}
     \mathbb{P}(H_{T})  = & 1 - \frac{1}{\pi} \left(\int_{\theta = 0}^{\frac{\pi}{2}} \int_{r_T = 0}^{\infty} f_{d_{0}}(r_T)\: e^{ -\lambda_T S_T } d r_T d\theta 
    \right.
     \nonumber\\&\left.
     + \int_{\theta = \frac{\pi}{2}}^{\pi} \int_{r_T = 0}^{v\:\mathrm{cos}(\pi - \theta)} f_{d_{0}}(r_T)\: e^{ -\lambda_T C_T }d r_T d\theta 
   \right.\nonumber\\&\left. + \int_{\theta = \frac{\pi}{2}}^{\pi} \int_{r_T = v\:\mathrm{cos}(\pi - \theta)}^{\infty}  f_{d_{0}}(r_T)\: e^{ -\lambda_T S_T }  d r_T d\theta \right). \nonumber
\end{align}
\end{corollary}}
Another special case is a situation when a mobile user moves in a straight line. In this case, \textbf{Lemma~2} can be simplified by substituting $\theta=0$ as follows:
\textcolor{black}{
\begin{corollary} Given a mobile user is initially associated with a TBS and moving in a straight line, the conditional probability of no HO from the serving TBS can be given as follows:
\begin{align}
    \mathbb{P}(\overline{H}_{T})  = & \frac{1}{\pi}  \int_{r_{T} = 0}^{\infty} f_{r_{T}}(r_{T}) \: \mathrm{exp} \left(-\lambda_{T} S_{T} - \lambda_{R} S_{T}^{\prime}\right) d r_{T} 
\end{align}
where
$ S_{T} =   \pi (R_{T}^{2} -  r_T^2)$,
$ S_{T}^{\prime} =  \pi R_{R}^{\prime 2} - r_{R}^{\prime 2}(\pi - \theta_{2}^{T}) + r_{R}^{\prime} \:v\: \mathrm{sin}{\theta_{2}^{T}}$, $ R_{T}^{2} = r_{T}^{2} + v^{2}+2r_{T}v,
 R_{R}^{\prime} = \left(R_{T}\right)^{\frac{2}{\alpha}}\:e^{ \frac{K_{a}\:R_{T}}{\alpha}} \left( \frac{P_{R}^{tx} Q}{P_{T}^{tx}} \right)^{\frac{1}{\alpha}}, r_{R}^{\prime} = \left(r_{T}\right)^{\frac{2}{\alpha}}\:e^{ \frac{K_{a}\:R_{T}}{\alpha}}\left( \frac{P_{R}^{tx} Q}{P_{T}^{tx}} \right)^{\frac{1}{\alpha}},  \theta_{2}^{T} = \mathrm{cos}^{-1} \left( \frac{r_{R}^{\prime 2} + v^2 - R_{R}^{\prime 2}}{2r_{R}^{\prime} v} \right)$.
\end{corollary}
}

\subsection{HO Probability Characterization from RBS}
\subsubsection{ HO Criterion from RBS}
Fig. \ref{graph}(b) denotes a situation where a mobile user is tagged with a given RBS at the position $c_2$. Given the user is associated to the RBS in a {multi-band} network, let $r_R$ be the distance between the mobile user and RBS whose PDF is given below. 
\begin{lemma}
\label{lm:Distance_distribution}
The PDF of the conditional distance $r_R$ from a typical mobile user initially connected to RBS to desired RBS can be acquired as follows:
\begin{align}
     f_{r_{R}}(r_{R}) = \frac{2\pi \lambda_{R} r_{R}}{A_{R}}\: \mathrm{exp} \left(-\pi \lambda_{R} r_{R}^2  -\pi \lambda_{T}\left(\frac{r_R^\alpha}{ Q}\right)^{\frac{2}{2+\mu}}\right).
\end{align}
\end{lemma}
\begin{proof}
See \textbf{Appendix~A}.
\end{proof}
The area centered by $c_2$ with radius $r_R$ is denoted by $B$.
Let $R_R$ denotes the distance between $c_1$ and RBS. When user moves to $c_1$, the HO occurs if the maximum received power of TBS is greater than that of  RF tier, i.e.,  $P_{R}^{\mathrm{rx}} < P_{T}^{\mathrm{rx}}$, which results in the following:
\begin{equation}\label{r2t_eqnv1}
    r_{T}^{2}\:\mathrm{exp}\left( K_{a}r_{T} \right) < (r_{R})^{\alpha} (1/Q).
\end{equation}
Note that the exponential term is a function of $r_{T}$; therefore, for the sake of tractability and to apply the equivalent distance approach, we approximate $r_{T}^{\mu} \approx \mathrm{exp}(K_{a}\:r_{T})$ then, $ r_{T}^{2}\:\mathrm{exp}(K_{a}\:r_{T}) \approx (r_{T})^{2+\mu} $.
Subsequently, we have
$
    r_{T} < \left[(r_{R})^{\alpha} (1/Q) \right]^{\frac{1}{2+\mu}} \triangleq r_{T}^{\prime},
$
where $\mu$ is a correcting factor and $r_{T}^{\prime}$ specifies the virtual distance of $r_{R}$. 
That is, when $r_{R}>r_{T}>r_{T}^{\prime}$, there will be no HO.

{\em \textbf{Choice of $\mu$}:} To select $\mu$ appropriately, we calculate the probability of association of the typical mobile user to RBS  by using the exact result of $A_T$ in \eqref{assoc1} and then equate it to the approximate association probability obtained as follows:
\begin{equation}
\label{eq:Asso_Probab_Proof}
\begin{split}
\tilde{A}_R & = \mathbb{E}_{r_0}\left[\mathbb{P}\left( P_{R}^{\mathrm{rx}} > P_{T}^{\mathrm{rx}} \right)\right],\\ 
& = \mathbb{E}_{r_0}\left[\mathbb{P}\left( P_{{R}}\gamma_{R} r_0^{-\alpha} > P_{{T}}\gamma_{T}\frac{\exp\left(-K_{a} d_0\right)}{d_0^2} \right)\right],\\ 
& \approx \mathbb{E}_{r_0}\left[\mathbb{P}\left( P_{{R}}\gamma_{R} r_0^{-\alpha}> P_{{T}}\gamma_{T}{d_0^{-2-\mu}} \right)\right],\\ 
& \stackrel{(a)}{=} \mathbb{E}_{r_0}\left[\exp \left(-\pi\lambda_{T}\left(\frac{r_0^\alpha}{Q}  \right)^{\frac{1}{2 +\mu}}\right)\right],\\ 
& {=} \int_{0}^{\infty} \exp \left(-\pi \lambda_{T}\left(\frac{r_0^\alpha}{Q}\right)^{\frac{2}{2+\mu}} \right) f_{r_0}(r_0)\:dr_0.
\end{split}
\end{equation}
Solving $A_R=\tilde{A}_R$ gives us the  appropriate value of $\mu$. Note that when the molecular absorption coefficient $K_a \rightarrow 0$, the value of $\mu \rightarrow 0$ and the approximation becomes exact, i.e., $\tilde A_R =1-A_T$.

\subsubsection{HO Analysis}  
When the typical mobile user is originally associated to THz tier, the HO probability from RBS can be determined by taking the average over $r_R$ and $\theta$ as follows:
\begin{align}
    \mathbb{P}(H_R) & = 1-\mathbb{P}(\overline{H}_{R})= 1-\mathbb{E}_{r_R,\theta}[\mathbb{P}(\overline{H}_{R}|r_R,\theta)].
\end{align}
\normalsize

To derive the no HO probability   ($\mathbb{P}[\bar{H}_{R}]$) of a mobile user  who is tagged to RBS, we replace the serving RBS by a virtual RBS in THz tier, therefore, the new distance is $r_{T}^{\prime}$ away from the target mobile user.  Fig. \ref{graph}(b) shows the area $B^{\prime}$ centered at $c_2$ with radius $r_{T}^{\prime}$. The fact that RF tier is acting as the serving tier to serve the mobile user indicates that all the other BSs of THz tier are situated outside $B^{\prime}$. Here, the area $A^{\prime}$ with the radius $R_{T}^{\prime}$, which is centered at $c_1$. The radius $R_{T}^{\prime}$ is the equivalent distance of $R_{R}$. If no TBSs are located in the area $A^{\prime}$, then initial RBS will remain the tagged BS even after movement of the user. There will be no HO if  no RBSs or TBSs are closer then  $r_{R}^{\prime}$. 
\begin{lemma} Given a mobile user is originally associated with a RBS, the probability of no HO from the serving RBS in a hybrid RF-THz network can  be derived as follows:
\begin{align}
    &\mathbb{P}[\bar{H}_{R}] =  \frac{1}{\pi} \left(\int_{\theta = 0}^{\frac{\pi}{2}} \int_{r_{R} = 0}^{\infty} f_{r_R}(r_R) e^{\left(-\lambda_{R} S_{R} - \lambda_{T} S_{R}^{\prime}\right)} dr_{R}\:d\theta \right.\nonumber\\
    &\left. + \int_{\theta = \frac{\pi}{2}}^{\pi} \int_{r_{R} = 0}^{v\:\mathrm{cos}(\pi - \theta)} f_{r_R}(r_R) e^{ \left(-\lambda_{R} C_{R} - \lambda_{T} C_{R}^{\prime}\right)} dr_{R}\:d\theta \right.\nonumber\\
    & \left.+ \int_{\theta = \frac{\pi}{2}}^{\pi} \int_{r_{R} = v\:\mathrm{cos}(\pi - \theta)}^{\infty} f_{r_R}(r_R) e^{ \left(-\lambda_{R} S_{R} - \lambda_{T} S_{R}^{'}\right)} dr_{R}\:d\theta
    \right),
    \nonumber
\end{align}
where $ S_{R} = |A| - |A\cap B|, S_{R}^{\prime} = |A^{\prime}| - |A^{\prime} \cap B^{\prime}|,  C_{R} = R_{R}^{2}(\pi - \theta_{1}^{\prime R}) + r_{R} \: v \: sin\theta - r_{R}^{2}(\pi - \theta),  C_{R}^{\prime} = R_{T}^{\prime 2}(\pi - \theta_{3}^{\prime R}) + r_{T}^{\prime} \: v \: \mathrm{sin}\theta_{2}^{R} - r_{T}^{\prime 2}(\pi - \theta_{2}^{R}), \theta_{1}^{\prime R} = \theta - \pi + \mathrm{sin}^{-1} \left(\frac{v\:\mathrm{sin}\theta}{R_R} \right), \theta_{3}^{\prime R} = \theta - \pi + \mathrm{sin}^{-1} \left(\frac{v\:\mathrm{sin}\theta}{R_{T}^{\prime}} \right)$.
\end{lemma}
\begin{proof}
See \textbf{Appendix~C}.
\end{proof}
Another special case is a situation when a mobile user moves in a straight line. In this case, \textbf{Lemma~4} can be simplified by substituting $\theta=0$ as follows:
\textcolor{black}{
\begin{corollary} Given a mobile user is initially associated with a RBS and moving in a straight line, the conditional probability of no HO from the serving RBS can be given as follows:
\begin{align}\label{hop_t2r}
    \mathbb{P}[\bar{H}_{R}] = & \frac{1}{\pi}\int_{r_{R} = 0}^{\infty} f_{r_R}(r_R) \: \mathrm{exp}\left(-\lambda_{R} S_{R} - \lambda_{T} S_{R}^{\prime}\right) dr_{R} 
\end{align}
where $ S_{R} = \pi (R_{R}^{2}- r_{R}^{2}), S_{R}^{\prime} =  \pi R_{T}^{\prime 2} - r_{T}^{\prime 2}(\pi - \theta_{2}^{R}) + r_{T}^{\prime}\:v \:\mathrm{sin}{\theta_{2}^{R}},
$
$R_{R}^{2} = r_{R}^{2} + v^{2}-2r_{R} \: v \: \mathrm{cos}(\pi - \theta) = r_{R}^{2} + v^{2}+2r_{R} \: v, R_{T}^{\prime}=\left[ (R_{R})^{\alpha}\:(\frac{P_{T}^{tx}}{P_{R}^{tx}\:Q}) \right]^{\frac{1}{2+\mu}}, r_{T}^{\prime}=\left[ (r_{R})^{\alpha}\:\left(\frac{P_{T}^{tx}}{P_{R}^{tx}\:Q}\right)\right]^{\frac{1}{2+\mu}}$, 
$\theta_{2}^{R} = \mathrm{cos}^{-1} \left( \frac{r_{T}^{\prime 2} + v^2 - R_{T}^{\prime 2}}{2r_{T}^{\prime} v} \right)$.
\end{corollary}
}

\subsection{Overall HO Probability}
The HO probability  of a mobile user in a hybrid RF-THz network finds as follows:
\begin{equation}\label{hop}
    \mathbb{P}(H)=  A_R\mathbb{P}(H_R)+ A_T\mathbb{P}(H_T)= 1-A_R\mathbb{P}(\bar H_R)- A_T\mathbb{P}(\bar H_T),
\end{equation}
where $A_R$ and $A_T$ are given in \eqref{assoc1}. Likewise, $\mathbb{P}(\bar H_T)$ and $\mathbb{P}(\bar H_R)$ are given by \textbf{Lemma~2} and \textbf{Lemma~4}, respectively.

\section{Coverage Probability With and Without Mobility}

In this section, first we  characterize the conditional coverage probabilities from TBS and RBS, i.e., $\mathbb{C_T}$ and $\mathbb{C}_R$, respectively. Then, the probability of coverage with and without mobility will be derived. Since a mobile user can connect to either RF or THz tier, the unconditional probability of coverage  without mobility can demonstrate as follows:
\begin{equation}
\label{eq:Total-Cov}
\mathbb{C} = A_T \mathbb{C}_T +A_R \mathbb{C}_R,
\end{equation}
where $A_T$ and $A_R$ are defined in Section~III. 
Here, the TBSs and RBSs are distributed as different PPPs, therefore, the distance of a typical mobile user to its serving BS depends on the associated tier. Afterwards, the PDF of the conditional distance  of the typical mobile user to TBS and RBS can be given as in {\bf Lemma~1} and {\bf Lemma~3}, respectively.  

\subsection{Conditional Coverage Probability - THz}
Conditioned on the fact that the mobile user is tagged to TBS,  the conditional rate coverage probability is defined as the probability of this user achieving a target data rate
 $R_{\mathrm{th}}$. Using  $R_{\mathrm{th}}=W_T\mathrm{log}_2(1+\mathrm{SINR}_{T})$ (where $W_T$ {is the bandwidth for THz transmission}), the conditional rate coverage probability is given below:
\begin{align}
\label{eq:cov_prob_def}
&\mathbb{C}_{T}= \mathbb{P}\left(\mathrm{SINR}_T>2^{\frac{R_{\mathrm{th}}}{W_T}}-1\right) = \mathbb{P}\left(\mathrm{SINR}_T>\tau_T\right), \nonumber \\
&= \mathbb{P}\left(\frac{P_{T}^{\mathrm{tx}} \gamma_{T} (r_{T,0})^{-2} \left((1+\tau_T)\: e^{\left(-K_{a}\:r_{T,0}\right)} - \tau_T\right)}{N_0 + \sum_{i\epsilon \Phi \backslash 0 } P_{T}^{\mathrm{tx}} \gamma_T F {(r_{T,i})^{-2}}}>\tau_T\right).
\end{align}
Taking 
$S(r_{T,0})=(1+\tau_T)P_{T}^{\mathrm{tx}} \gamma_{T} (r_{T,0})^{-2}\: \mathrm{exp}\left(-K_{a}r_{T,0}\right) - P_{T}^{\mathrm{tx}} \gamma_T (r_{T,0})^{-2} \tau_T,$ 
where, $I_{T}^{\mathrm{agg}}=\sum_{i\epsilon \phi \backslash 0 } P_{T}^{\mathrm{tx}} \gamma_T F {(r_{T,i})^{-2}}$ is the cumulative interference at the typical mobile user and applying Gil-Pelaez inversion theorem, Eq. (\ref{eq:cov_prob_def}) can be rewritten as follows:
\begin{align} \label{eq:cov_prob_def_2ndStep}
    \mathbb{C}_T & = \mathbb{P}\left( \frac{S(r_{T,0})}{N_0 + I_{T}^{\mathrm{agg}}}>\tau_T \right), \nonumber\\
    & = \mathbb{P}\left( S(r_{T,0})>\tau_T N_0 +\tau_T I_{T}^{\mathrm{agg}}\right), \nonumber\\
    & = \mathbb{E}_{r_{T,0}}\left[ \frac{1}{2} -\frac{1}{\pi} \int_{0}^{\infty} \frac{\mathrm{Im}[\phi_{\Omega|r_{T,0}}(\omega) \: \mathrm{exp}\left(j\omega\tau_T N_0 \right)]}{\omega} d\omega \right], 
\end{align}
where Im$(\cdot)$  {denotes} the imaginary operator, $\Omega=S(r_{T,0})-\tau_{T} I_{\mathrm{agg}}^{T}$, and $\phi_{\Omega}(\omega)=\mathbb{E}[ \mathrm{exp}\left(-j\omega \Omega\right)]$  {denotes} the characteristic fuction (CF) of $\Omega$ {can be stated} as follows:
\begin{align}
    \phi_{\Omega|r_{T,0}}(\omega)=\mathrm{exp}\left(-j\omega S(r_{T,0})\right) \mathcal{L}_{I_{T}^{\mathrm{agg}}|r_{T,0}}(-j\omega \tau_T).
\end{align}
where $\mathcal{L}_{I_{\mathrm{agg}}^{T}|r_T}$ is the LT of the cumulative interference conditioned on $r_{T,0}$ and is derived in the following Lemma. 

\begin{lemma} The LT of the cumulative interference can be given as follows:
\begin{equation}
    \mathcal{L}_{I_{T}^{\mathrm{agg}}}(s)
= \exp\Biggl(2\pi\lambda_{T}\sum_{l=1+\epsilon}^{\infty}\frac{\left(-sF \gamma_{T} P_{\mathrm{T}}^{\mathrm{tx}}\right)^{l}}{\left(2l-2\right)l!}\cdot \frac{1}{(r_{T,0})^{2l-2}} \Biggr), \nonumber
\end{equation}
where \textcolor{black}{$F=(\theta_{\mathrm{tx}}\:\theta_{\mathrm{rx}})/4\pi^2$} is the probability of main-lobe alignment of the interferers and the typical user and given negligible side lobe gains.
\end{lemma}
\begin{proof}
Starting from the definition of LT, we have:
\begin{equation*}
\label{eq:Laplace_Trans-Proof} 
\begin{split} 
&\mathcal{L}_{I_{T}^{\mathrm{agg}}}(s)  = \mathbb{E}_{\Phi_{T}}\left[\mathrm{exp}\left(-s I_{T}^{\mathrm{agg}}\right)\right] 
\\&
= \mathbb{E}_{\Phi_{T}}\Biggl[\mathrm{exp}\left(-sF \sum_{i \in \Phi_{T} \backslash 0 } P_{T}^{\mathrm{tx}} \gamma_T {(r_{T,i})^{-2}} \right)\Biggr],
\\ 
& = \mathbb{E}_{\Phi_{T}}\Biggl[\prod_{i \in \Phi_{T}\backslash 0}\exp\biggl(-sF P_{T}^{\mathrm{tx}} \gamma_{T} (r_{T,i})^{-2}\Biggr)\Biggr],
\\
& \stackrel{(a)}{=} \exp\Biggl(-2\pi\lambda_{T}\int_{r_{T,0}}^{\infty}r_{T,i}\left(1- \exp{\left(-sF \gamma_{T} P_{\mathrm{T}}^{\mathrm{tx}} (r_{T,i})^{-2}\right)}\right)d r_{T,i}\Biggr),
\\&\stackrel{(b)}{=}
\exp\Biggl(2\pi\lambda_{T}\int_{r_{T,0}}^{\infty}\sum_{l=1+\epsilon}^{\infty}\frac{\left(-sF \gamma_{T} P_{\mathrm{T}}^{\mathrm{tx}} \right)^l }{(r_{T,i})^{2l-1} l!} 
dr_{T,i}\Biggr), \\ 
&\stackrel{(c)}{=}\exp\Biggl(2\pi\lambda_{T}\sum_{l=1+\epsilon}^{\infty}\frac{\left(-sF \gamma_{T} P_{\mathrm{T}}^{\mathrm{tx}}\right)^{l}}{\left(2l-2\right)l!}\: \frac{1}{(r_{T,0})^{2l-2}} \Biggr),
\end{split}
\end{equation*}
\normalsize
where  {using probability generating functional (\textbf{PGFL})  $f(x)=\exp\left(-sP_{T}h\left(r_{T,i}\right)\right)$ step (a) has obtained,}
(b) is {obtained} by using $\exp\left(-x\right)=\sum_{i=0}^{\infty}(-1)^{i}\frac{x^{i}}{i!}$ (\cite{gradshteyn2014table}, Eq. 1.211) and $\epsilon =0.01$ is inserted to avoid the indeterminate term. Since the {mobile} user has {maintained} a distance $r_{T,0}$ from its {tagged} TBS, all interferers are beyond   $r_{T,0}$, which is  the lower limit of the integral. 
\end{proof}
\subsection{Conditional Coverage Probability - RF}
The conditional probability of coverage of the {mobile} user, which is tagged to RBS can be derived in an interference limited regime  as follows \cite{9119462}:
\begin{align}\label{CCRF}
    \mathbb{C}_{R} & = \mathbb{P} \left(\frac{P_{R}^{tx} \gamma_{R} H}{(r_{R,0})^{\alpha} I_{R}^{\mathrm{agg}}} > \tau_{R} \right) \nonumber\\
    &= \mathbb{P} \left(H > \tau_{R} (P_{R}^{tx})^{-1} \gamma_{R}^{-1} (r_{R,0})^{\alpha} I_{R}^{\mathrm{agg}} \right),  \nonumber\\
    & = \mathbb{E}\left[ \exp(-\tau_{R} (P_{R}^{tx})^{-1} \gamma_{R}^{-1} (r_{R,0})^{\alpha} I_{R}^{\mathrm{agg}}) \right], \nonumber\\
    & = \int_{0}^{\infty} \mathcal{L}_{I_{R}^{\mathrm{agg}}} \left( \tau_{R} (P_{R}^{tx})^{-1} \gamma_{R}^{-1} (r_{R,0})^{\alpha} \right) f_{r_{R,0}}(r_{R,0}) d r_{R},
\end{align}
where $ \mathcal{L}_{I_{R}^{\mathrm{agg}}}\left( \tau_{R} (P_{R}^{tx})^{-1} \gamma_{R}^{-1} (r_{R,0})^{\alpha} \right) = \mathrm{exp}\left( -\pi (r_{R,0})^{2} \lambda_R \mathcal{Y}(\tau_R, \alpha)\right)$. Here, $\mathcal{Y}(\tau_R, \alpha) = \frac{2\tau_R}{\alpha-2} {}_2 F_1 [1,1-\frac{2}{\alpha};2-\frac{2}{\alpha};-\tau_R],$ and $ {}_2 F_1[\cdot] $ is the Gauss's Hypergeometric function.

\subsection{Coverage Probability With and Without Mobility}
The overall coverage probability without mobility in a hybrid RF-THz network is given by substituting the conditional coverage probability results given in \eqref{eq:cov_prob_def_2ndStep} and \eqref{CCRF} into \eqref{eq:Total-Cov}.
The overall coverage probability with mobility is a function of HO probability. 
The coverage degrades with the higher HO probability, service delays, and dropped calls. The overall coverage with mobility can then be modeled as follows~\cite{9023398}: 
\begin{align} \label{overall_CP}
    \mathbb{C}_M  = \mathbb{C} \left(1 - \eta \: \mathbb{P}(H) \right),
\end{align}
\textcolor{black}{The coefficient $\eta$, in effect, measures the system sensitivity to HOs. Its value depends on a number of factors, e.g., the radio access technology, the mobility protocol, the protocol’s layer of operation and the link speed. At one extreme, as $\eta \rightarrow 0$,  there is no HO cost and system HO failures do not happen. On the other hand, as $\eta \rightarrow 1$, every HO results in an outage.} Note that $\mathbb{C}$ is the total coverage probability of a typical {mobile} user in a hybrid RF-THz network without mobility given by \eqref{eq:Total-Cov}  and $\mathbb{P}(H)$ is the overall HO probability {for that mobile} user in a hybrid RF-THz network given by \eqref{hop}. 

{\bf Remark:} As a special case, the overall coverage probability without mobility in a stand-alone THz network can be given simply by averaging \eqref{eq:cov_prob_def_2ndStep} over $d_0$ instead of $r_{T,0}$. 

Furthermore, in the case of noise-limited regime, i.e., when the interference is negligible, the coverage probability in a stand-alone THz network can be simplified as shown in the following.

\begin{corollary} In the noise-limited regime (in scenarios where the intensity of TBSs is low), the coverage probability of a typical user can be simplified as follows:
\begin{align} \label{eq:cov_prob_def_2ndStep}
    &\mathbb{C}_T 
     = \mathbb{P}\left( S(r_{T,0})>\tau_T N_0 \right), \nonumber\\
    & \stackrel{(a)}{=} \mathbb{E}_{r_{T,0}}\left[ \frac{1}{2} -\frac{1}{\pi} \int_{0}^{\infty} \frac{\mathrm{Im}[\mathrm{exp}\left(-j\omega (S(r_{T,0})-\tau_T N_0 \right)]}{\omega} d\omega \right],  \nonumber\\
    & =   \frac{1}{2} +\frac{1}{\pi} \mathbb{E}_{r_{T,0}}[\int_{0}^{\infty} \frac{\mathrm{sin}\left(\omega (S(r_{T,0})-\tau_T N_0 \right)}{\omega} d\omega] ,
\end{align}
where (a) follows from Euler's identity, i.e., $e^{-j \theta} =\mathrm{cos}\theta -j \mathrm{sin} \theta$.
\end{corollary}

\subsection{Extension to Incorporate Misalignment}
\textcolor{black}{The misalignment errors in the desired signal can be incorporated with the path-loss as a normal random variable \cite{ekti2017statistical} with zero mean and finite variance. Let the LT of misalignment variable $\chi$  is given by $\mathcal{L}_\chi(\cdot)$, we can update the coverage probability calculation  as follows: 
  \begin{align} \label{eq:cov_prob_def_2ndStep}
    &\mathbb{C}_T  = \mathbb{P}\left( \frac{S(r_{T,0}) \chi}{N_0 + I_{T}^{\mathrm{agg}}}>\tau_T \right)
     = \mathbb{P}\left( S(r_{T,0} ) \chi >\tau_T N_0 +\tau_T I_{T}^{\mathrm{agg}}\right), \nonumber\\
    & = \mathbb{E}_{r_{T,0}}\left[ \frac{1}{2} -\frac{1}{\pi} \int_{0}^{\infty} \frac{\mathrm{Im}[\phi_{\Omega|r_{T,0}}(\omega) \: \mathrm{exp}\left(j\omega\tau_T N_0 \right)]}{\omega} d\omega \right], 
\end{align}
 where $\Omega=S(r_{T,0}) \chi -\tau_{T} I_{\mathrm{agg}}^{T}$ and 
$
    \phi_{\Omega|r_{T,0}}(\omega)=
    \mathcal{L}_{\chi}(j\omega S(r_{T,0}))
    \mathcal{L}_{I_{T}^{\mathrm{agg}}|r_{T,0}}(-j\omega \tau_T).
$
}
\subsection{Extension to Incorporate Blockages}
\textcolor{black}{Along the lines of \cite{6840343}, the blockages can be modeled as a thinning process in the stochastic geometry models where the transmissions from a specific number of BSs are considered as blocked.    Similar to \cite{6840343}, we consider a Boolean blockage model in which obstacles are rectangles and are distributed following a homogeneous PPP of density $\lambda_{B}$. The rectangles length ($L_{k}$) and width ($W_{k}$) are independent and identically distributed, and their probability density functions are $f_{L}(x)$ and $f_{W}(x)$, respectively. The orientation of these rectangles is distributed uniformly in $\left[0,2\pi\right)$. Then, according to \cite{6840343}, the number of blockages in a link of length $r_{T,0}$ is a random variable with Poisson distribution having mean $\xi r_{T,0}+p$, where $\xi = \frac{2\lambda_{B}\left(E\{W\}+E\{L\}\right)}{\pi}$ and $p = \lambda_{B}E\{W\}E\{L\}$, where $0<p<1$ represents the area which is under blockages. Thus, the LOS probability is $p_{\mathrm{LOS}}(r_{T,0}) = \exp{\left[-\left(\xi r_{T,0} + p\right)\right]}$ and this factor can be multiplied with the coverage probability  to consider the impact of  blockages.
 Subsequently, the coverage probability $\mathbb{C}_T$ now depends on two events, i.e, \textbf{(i)} the LOS link is not blocked; and \textbf{(ii)} $\mathrm{SINR}_T$ is greater than $\tau_T$. For a given distance $r_{T,0}$,  $\mathbb{C}_T$ can be expressed as in \eqref{eq:cov_prob_def_2ndStep},
 \begin{figure*}
   \begin{align} \label{eq:cov_prob_def_2ndStep}
    \mathbb{C}_T & = p_{\mathrm{LOS}}(r_{T,0}) \, \mathbb{P}\left( \frac{S(r_{T,0})}{N_0 + I_{T}^{\mathrm{agg}}}>\tau_T \right), \nonumber\\
    & = \mathbb{E}_{r_{T,0}}\left[ \frac{p_{\mathrm{LOS}}(r_{T,0})}{2} -\frac{1}{\pi} \int_{0}^{\infty} p_{\mathrm{LOS}}(r_{T,0}) \frac{\mathrm{Im}[\phi_{\Omega|r_{T,0}}(\omega) \: \mathrm{exp}\left(j\omega\tau_T N_0 \right)]}{\omega} d\omega \right], 
\end{align}
\hrule
\end{figure*}
 where $\Omega=S(r_{T,0}) -\tau_{T} I_{\mathrm{agg}}^{T}$.}
 
 \subsection{Extension to Incorporate Ping-Pong Effect}
 \textcolor{black}{Typically, in wireless networks, the handover is initiated if the received signal power of another BS becomes better than the current serving BS by at least by a predefined  factor.  To consider the ping-pong effect, we can incorporate a  constant bias (hysteresis factor) $\eta_H$, i.e., by scaling the received signal power of the serving BS (whether RF or THz BS)  in Section-II.C in a straight-forward manner. For higher values of $\eta_H$, the handoff probability (and thus ping-pong effect) will reduce. However, optimizing this factor is beyond the scope of this article and is an interesting future research direction. }
  
\section{Numerical Results and Discussions}
In this section, {the derived expressions has been validated by} Monte-Carlo simulations {with the consideration of} a hybrid RF-THz network. Our results extract useful insights related to the probability {of coverage} of a moving user {within} a hybrid RF-THz network considering the impact of molecular noise in THz transmission, intensity of TBSs, desired rate requirement, and velocity of a user. \textcolor{black}{We have a general mobility model where the typical user can have any arbitrary trajectory with random distances and directions at each movement step. }

\subsection{Simulation Parameters}
Unless stated otherwise, the  users and BSs are located within a radius of 500 m circular region. The transmit power from TBS is 0.2~W and intensity of TBS is 0.0001 $\mathrm{BSs/ m^{2}}$. The transmit and receive antenna gains (i.e., $G_{T}^{\mathrm{tx}}$ and $G_{T}^{\mathrm{rx}}$) of TBSs are considered as 25 $\mathrm{dB}$.  Desired target rate  is 1 $\mathrm{Gbps}$ and THz transmission bandwidth is taken as 0.5 $\mathrm{GHz}$.  On the other hand, the transmit power from RBS is 2 W, where its transmission frequency is 2 $\mathrm{GHz}$ and transmission bandwidth is 40 $\mathrm{MHz}$. Here, {the exponent of} path loss  $\alpha $ is 4, and the intensity of RBSs is 0.00001 $\mathrm{ BSs/ m^{2}}$. \textcolor{black}{First three terms of Lemma~5 provide a good approximation and are used to compute numerical results.}
The simulation parameters to compute $K_a(f)$ are listed in Table~II.

\begin{table*}
\footnotesize
  \centering
  \caption{Simulation Parameters for Calculating $K_a(f)$  \cite{5995306}}
    \label{MOLECULAR ABSORPTION PARAMETERS}
    \begin{tabular}{| p{2.375cm} | p{3.7cm} | p{2cm} |p{4.5cm} |}
     \hline 
     \textbf{Symbol}  &\textbf{Value}   &\textbf{Symbol}& \textbf{Value}
    \\\hline
    $p_0, p$  &1 atm, 1 atm & $q^{(i,g)}$&0.05 [\%]\\\hline
    $T_0,T$ & 296~K, 396~K  &$k_b$ &1.3806$\times 10^{-23}$ J/K\\ \hline
    $f_{ {c_0}}^{(i,g)}$ &276 Hz & $T_{\mathrm{sp}}$ &273.15 K\\\hline
    $\gamma$ &0.83 & $N_A$ & 6.0221 $\times 10^{23}$\\   \hline
    $S^{(i,g)}$ &2.66$^{-25}$Hz-m$^2/$mol & h &6.6262$\times 10 ^{-34}$ J s\\\hline
    $\alpha_0^{(i,g)}, \alpha_{\mathrm{air}}^{(i,g)}$ &0.916Hz, 0.1117Hz & c&2.9979 $\times 10^8$ m/s \\\hline
    ${\delta}^{{(i,g)}}$ &0.0251 Hz & V &8.2051$\times 10^{-5}$m$^3$atm/K/mol\\
    \hline
  \end{tabular}
\end{table*}
\

\begin{figure*}
\centering
\begin{tabular}{lccccc}
\includegraphics[width=0.5\textwidth]{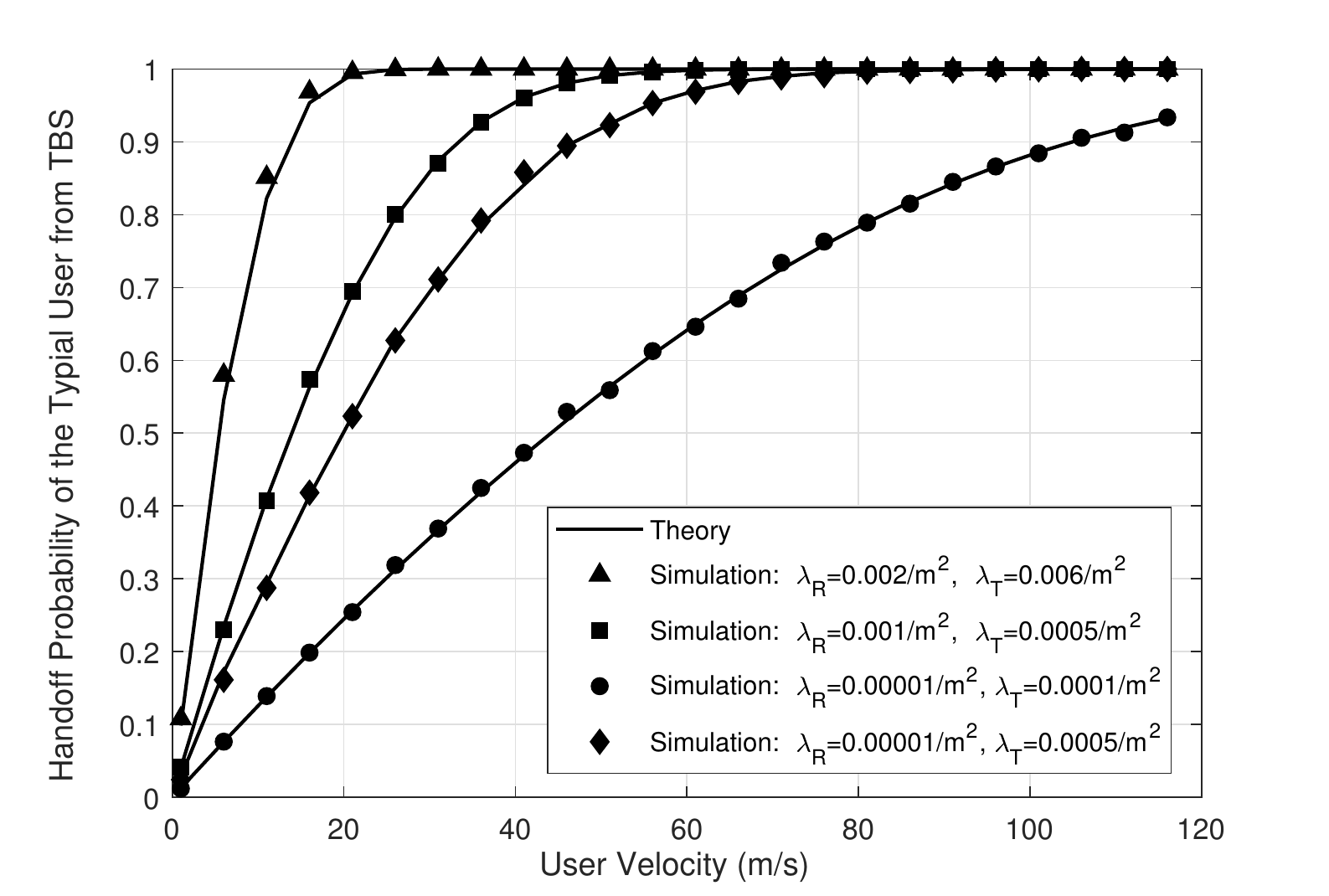}
&\hspace{-1cm}
{\includegraphics[width=0.5\textwidth]{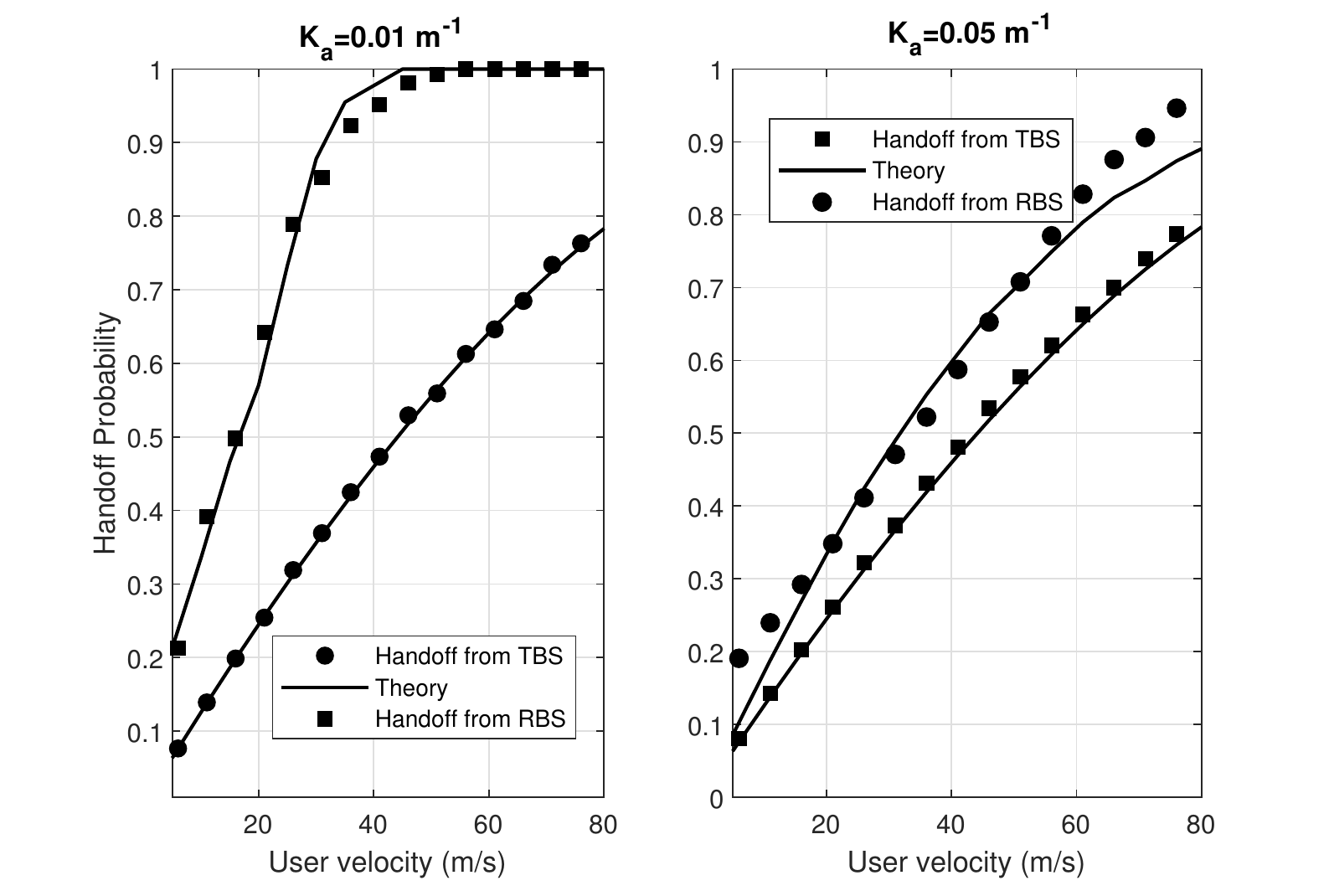}\label{fig:f3}}
&\hspace{-1cm}
\qquad \qquad\qquad \qquad (a)  & (b) 
\end{tabular}
\caption {(a) HO probability of a typical user from TBS as a function of the velocity and intensity of TBSs, $K_a=0.01$m$^{-1}$. (b) HO probability of a typical user as a function of the velocity and molecular absorption in a hybrid RF-THz network, $\lambda_R=0.00001$ per m$^{2}$ and $\lambda_T=0.0001$ per m$^{2}$.} 
\label{fig1}
\end{figure*}

\begin{figure}[!tbp] 
  \centering
{\includegraphics[width=0.5\textwidth]{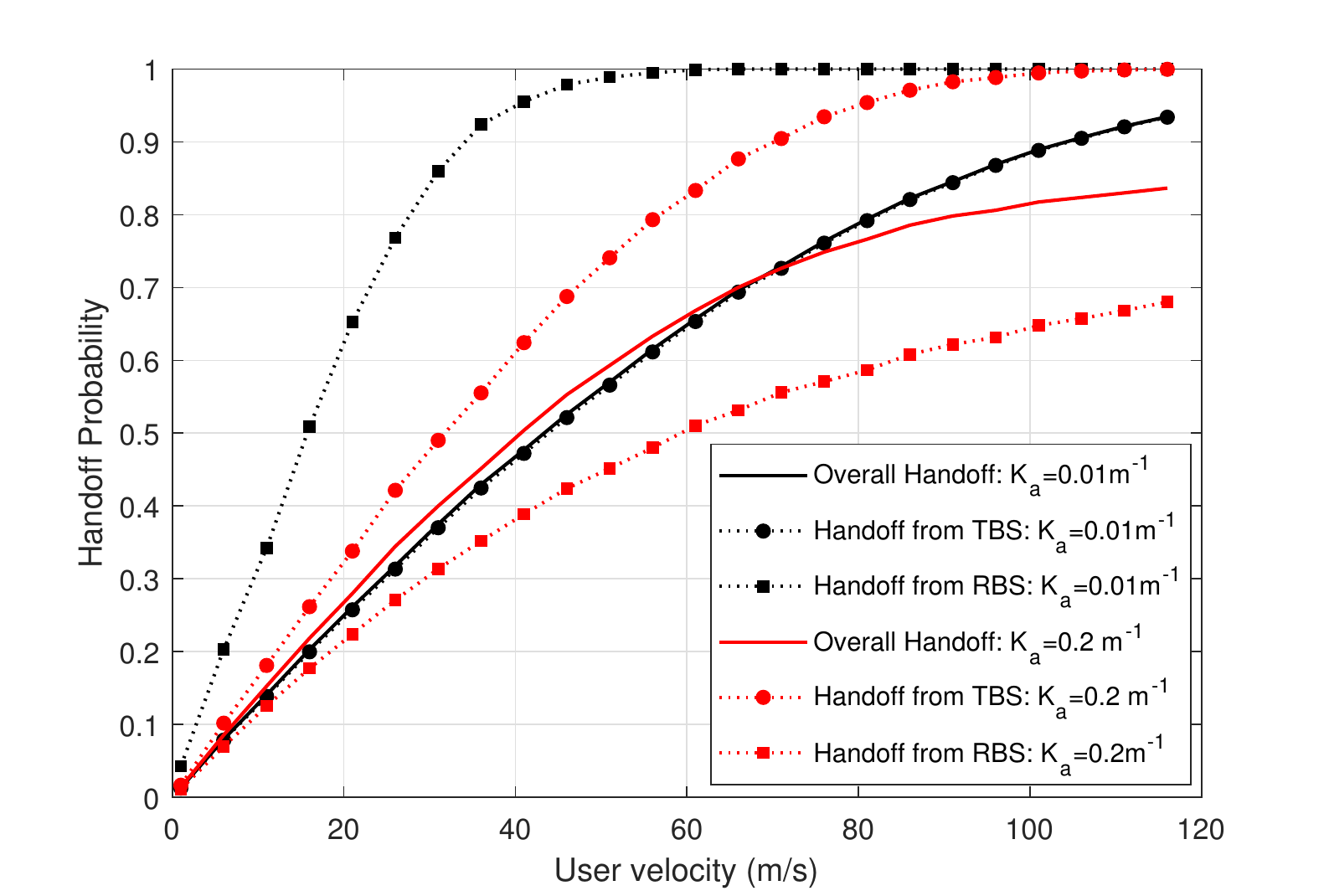}\label{fig:f2}}
  \caption{Overall HO probability of a typical user as a function of the velocity and molecular absorption in a hybrid RF-THz network, $\lambda_R=0.00001$ per m$^{2}$ and $\lambda_T=0.0001$ per m$^{2}$.}
  \label{sim_theory2}
\end{figure}

\subsection{Results and Discussions}
Fig.~3(a) depicts the HO probability {calculated for} a typical {mobile} user who is {tagged} to TBS  as a function of its velocity and intensity of {the} RBSs and TBSs. We compare the accuracy of our numerical results with the corresponding simulation results and note that our analytical results match perfectly with the Monte-Carlo simulations. The molecular absorption coefficient is set as $K_a = 0.01$m$^{-1}$. We observe that the HO probability increases with the increase in  the number of RBSs; however, the increase is much more significant with the increase in TBSs. That is, we note that increasing $\lambda_T$ from  $ 0.0001$ to $ 0.0005$ per m$^2$ [i.e., almost 5 times] at   $\lambda_R = 0.00001$ per m$^2$  increases the HO probability much more compared to the case when $\lambda_R$ increases from  $ 0.00001$ to $ 0.001$ per m$^2$ [i.e., almost 100 times] at   $\lambda_T = 0.0005$ per m$^2$. This signifies the impact of severe molecular absorption and small coverage zones on the connectivity of mobile users transmitting  in a hybrid RF-THz network. 

{Fig.~3(b) depicts the impact of user's velocity on the conditional HO probability from TBS and conditional HO probability from RBS considering  two different absorption coefficients, i.e., $K_a = 0.01$ m$^{-1}$ and $K_a = 0.05$ m$^{-1}$.  For lower values of $K_a$, the  HO probability is much higher if the user was initially associated to RBS compared to the case if the user was initially connected to TBS. Also, in this case, the HO probability increases much rapidly with the increase in velocity. The reason is the high-received signal power from TBS compared to RBS generates more BS-switching.  Conversely, for higher molecular absorption (i.e., $K_a = 0.05$ m$^{-1}$), we note an opposite trend. That is, the  HO probability is much higher if the user was initially associated to  TBS compared to the case when the user was initially connected to RBS. The reason is the severe molecular absorption which degrades the received signal from TBS and favors shifting users from TBS to RBS.}
Finally, it can be observed that as the molecular absorption coefficient $K_a \rightarrow 0$, our analytical results matches perfectly with the Monte-Carlo simulations. Conversely, for of $K_a=0.05$, the impact of approximation can be observed clearly. Note that the HO from TBS is exact and the approximation is only in the HO from RBS.

Fig. \ref{sim_theory2} has been plotted by keeping the intensity of BSs fixed  (i.e., $\lambda_R = 0.00001$ per m$^2$ and $\lambda_T = 0.0001$ per m$^2$) and tracking the variation of the HO probability from TBS ($H_T$), HO probability from RBS  ($H_R$), and overall HO probability of a typical mobile user with different molecular absorption coefficients. In the figure, the black curves consider $K_a = 0.01$~m$^{-1}$, whereas the red curves assume  $K_a = 0.2$ m$^{-1}$. It can be observed that for low molecular absorption coefficient, the  HO probability from RBS is much higher. The reason is that the lower molecular absorption favours association with TBS due to higher received powers.   Conversely, when the molecular absorption coefficient is high,  the  HO probability from TBS is much higher, i.e., for the user who is initially associated with TBS.  The reason is that the impact of molecular absorption is devastating; therefore, the criterion favours association with RBS.  Nevertheless, the overall HO probability remains nearly the same which highlights the significance of computing and extracting insights from ($H_T$) and  ($H_R$) separately.


\begin{figure}[tbp!]
  \centering
  {\includegraphics[width=0.5\textwidth]{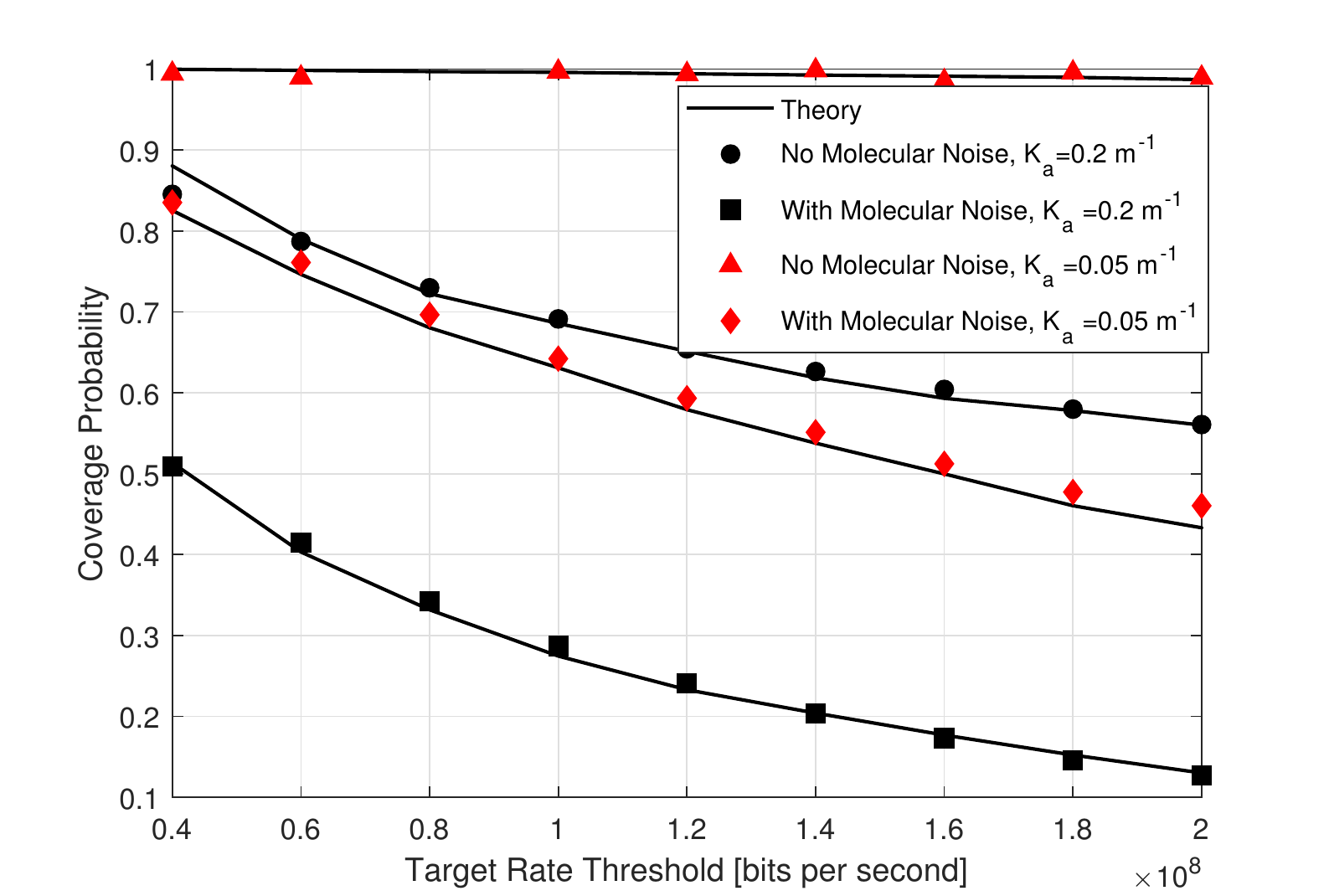}\label{fig:f4}}
  \caption{Coverage probability of a typical user as a function of the  target rate threshold  (in bps), $\lambda_R=0.00001$ per m$^{2}$ and $\lambda_T=0.0001$ per m$^{2}$.}
  \label{CPvsRate}
\end{figure}

\begin{figure}[tbp!]
  \centering
  {\includegraphics[width=0.5\textwidth]{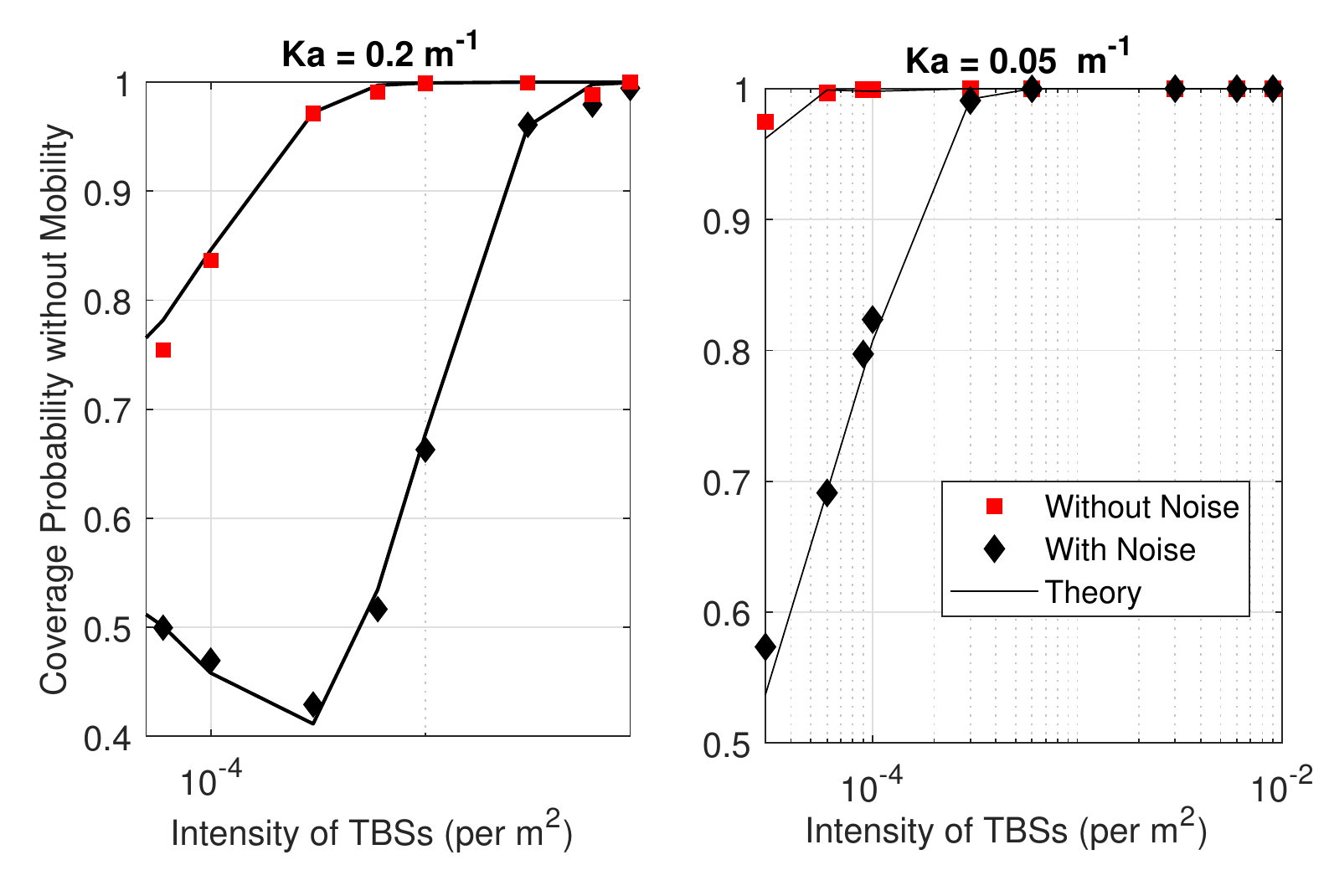}\label{fig:f6}}
  \caption{Coverage probability of a typical user as a function of the intensity of TBSs and molecular absorption in a hybrid RF-THz network, $K_a=0.05$m$^{-1}$ and  $K_a=0.01$m$^{-1}$.}
  \label{combine_CP_intensity}
\end{figure}

\begin{figure}[tbp!]
  \centering
  {\includegraphics[width=0.5\textwidth]{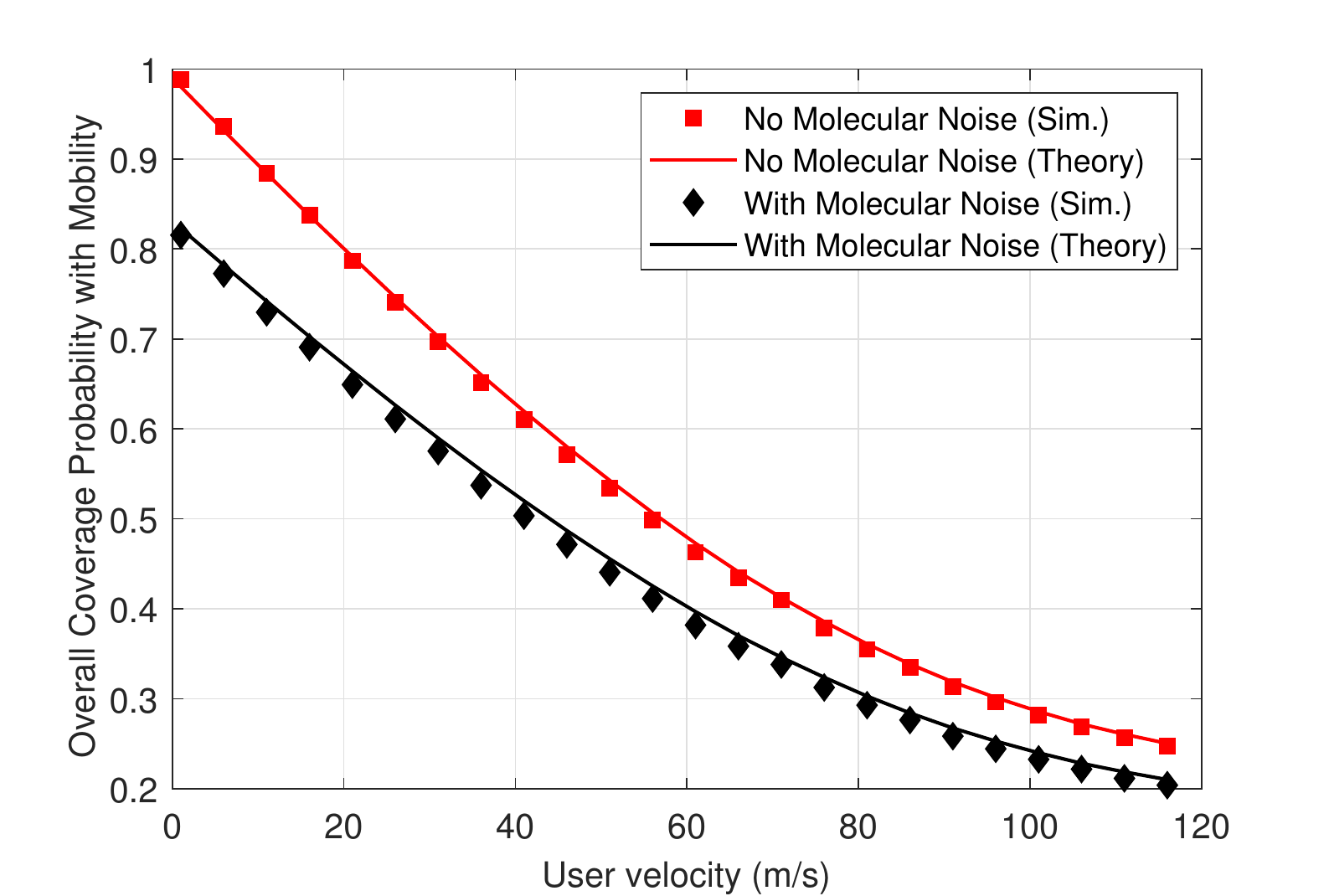}\label{fig:f5}}
  \caption{Mobility-aware coverage probability {with a relation of} a function of the user's velocity with and without molecular absorption, $\mu=0.82$, $K_a=0.05$m$^{-1}$.}
  \label{combine_CP_mobility}
\end{figure}

Fig. \ref{CPvsRate} demonstrates the impact of thermal and molecular absorption noise on the coverage probability of a static user for two different molecular absorption coefficients, $K_a$, as a function of user's desired data rate. Our analytical results match well with the simulation results.  As expected, the coverage probability decreases with the increase in the target data rate. Furthermore, the impact of molecular noise is devastating and substantiate that ignoring molecular noise from the analytical results (as is done in \cite{9119462}) can lead to over-optimistic results. Furthermore, with the increase in molecular absorption, the coverage probability degrades considerably. 

Fig. \ref{combine_CP_intensity} depicts the coverage probability with and without  molecular noise as a function of the intensity of TBSs. The intensity of RBSs kept constant, and  the user is static, however, the  coverage probability is observed for two different values of molecular absorption coefficient. The analytical results corroborate with  the simulation results.  This figure also confirms that the  molecular noise  significantly degrades the coverage probability compared to the case when there is no molecular noise \cite{9119462}.  Furthermore, in general, the increase in intensity of TBSs increases coverage due to the shortening of distance from the nearest TBS. Interestingly,  when molecular absorption is high, increasing the intensity  first deteriorates the coverage {probability} (due to increased interference); and afterwards escalate due to improved signal quality which is mainly due to shorter distance from the associated TBS.


Fig. \ref{combine_CP_mobility}, demonstrate the overall coverage probability with mobility from Eq. (\ref{overall_CP})  as a function of user's velocity. Numerical values from simulation results validate the accuracy of our theoretical results. 
This figure confirms the overall coverage probability reduces with the increase in velocity and demonstrates the gap between the results with  molecular noise and without molecular noise in \cite{9119462}.

\begin{figure}[tbp!]
  \centering
  {\includegraphics[width=0.5\textwidth]{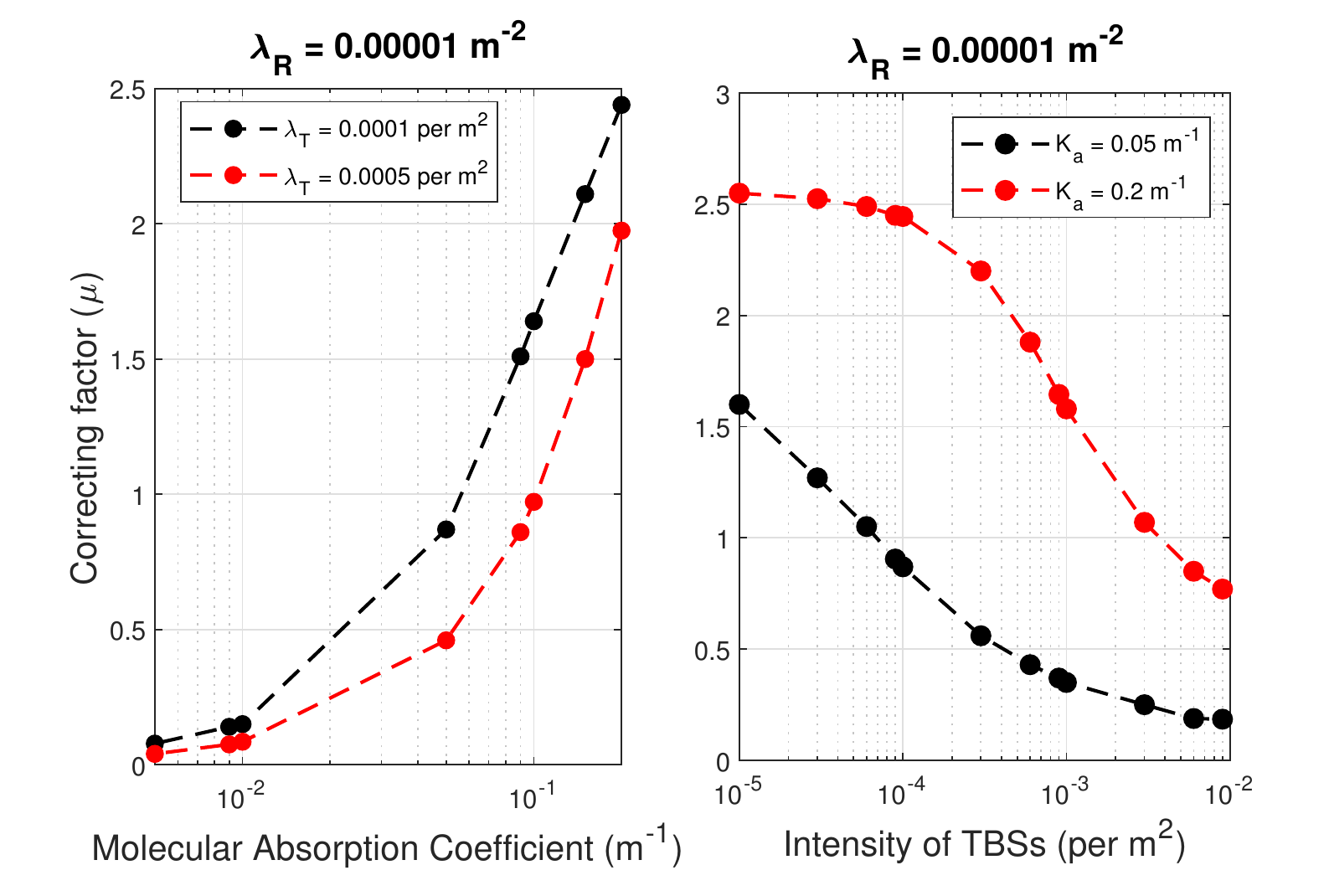}\label{fig:f9}}
  \caption{{The correction factor as a function of the intensity of {the} TBSs and molecular absorption.}}
  \label{ka_lambdaVSmu}
\end{figure}

\begin{figure}[tbp!]
  \centering
  {\includegraphics[width=0.5\textwidth]{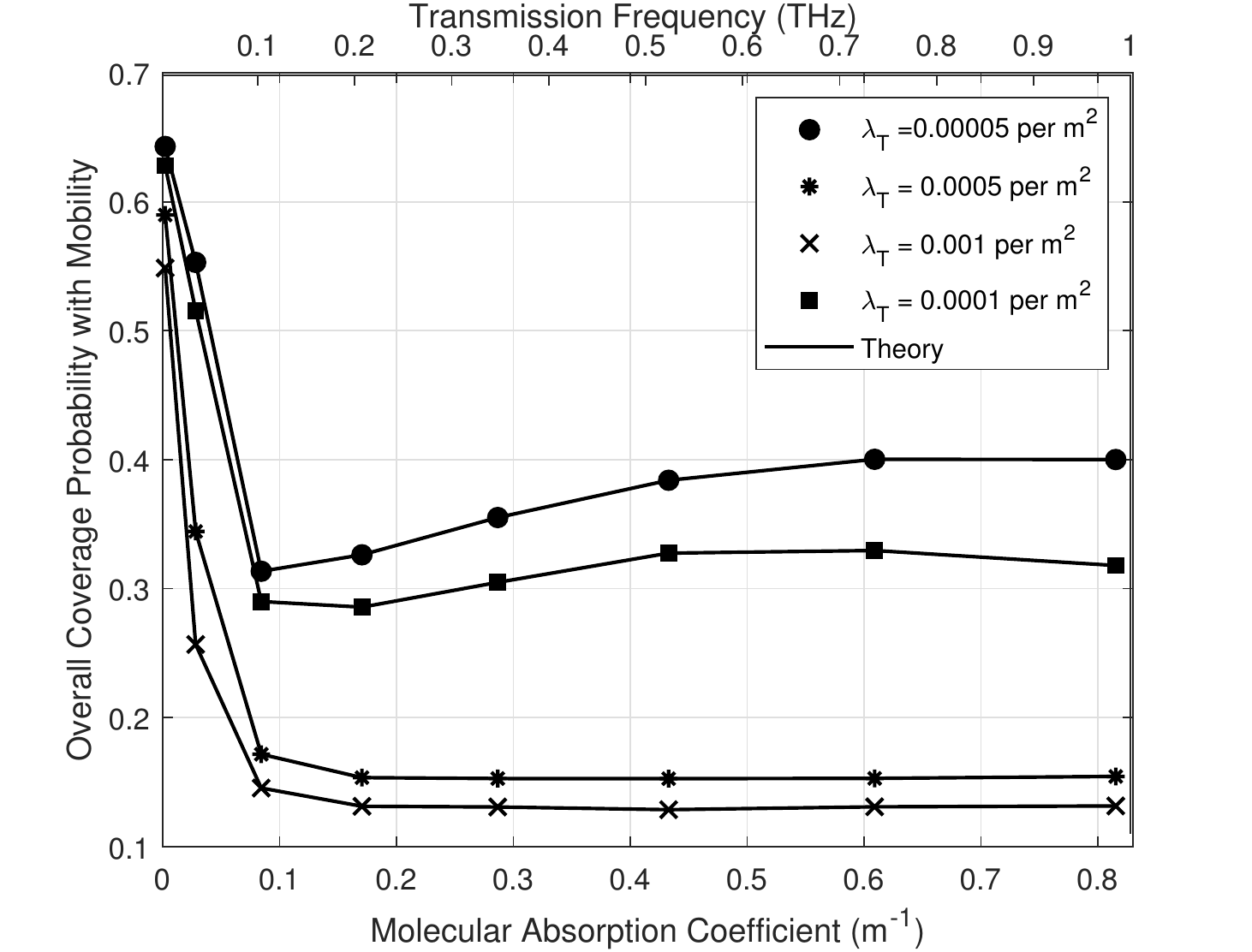}\label{fig:f9}}
  \caption{{Overall coverage probability of a typical user with mobility  as a function of the molecular absorption coefficients and frequencies considering  molecular noise and a constant user's velocity.}}
  \label{kaVScp}
\end{figure}

Fig. \ref{ka_lambdaVSmu} illustrates the correcting factor ($\mu$) as a function of the molecular absorption coefficient ($K_a$) and intensity of TBSs considering a fixed intensity of RBSs. The figure shows that the correcting factor gradually increases with the increasing value of molecular absorption coefficient. On the other hand, the figure demonstrates that the correcting factor gradually decreases with the increasing  intensities of TBSs.  

Fig.~\ref{kaVScp} depicts the overall coverage probability with mobility as a function of the molecular absorption coefficient  and intensity of TBSs. We consider the user is moving with a constant velocity (i.e., 56 m/s). The higher intensity of TBSs ($\lambda_T=0.001$ per m$^2$ and $\lambda_T=0.0005$ per m$^2$) results in a much denser THz network and the connected TBSs are likely to be much closer to the moving user. Therefore, increasing {value of} molecular absorption simply degrades the {probability of coverage}. This degradation is due to the reduction in signal strength. However, a more interesting observation can be noted from the cases when $\lambda_T=0.0001$ per m$^2$ and $\lambda_T=0.00005$ per m$^2$. That is, the overall coverage probability first reduces due to the  signal degradation as a function of the molecular absorption. However,  the coverage starts increasing again at some point and the reason is the reduction in interference with the increase of molecular absorption coefficient. Surprisingly, the benefits of interference reduction due to increasing molecular absorption dominates the drawback  of signal degradation for a reasonable intensity of TBSs. This trend is opposite to what observed for much denser THz network.

Finally, Fig.~\ref{blockage} denotes the impact of  blockage intensity $\lambda_B$ or blockage area $p$ on the coverage probability. As expected, increasing either  $\lambda_B$ or $p$ results in the coverage degradation.

\begin{figure}[tbp!]
  \centering
  {\includegraphics[width=0.5\textwidth]{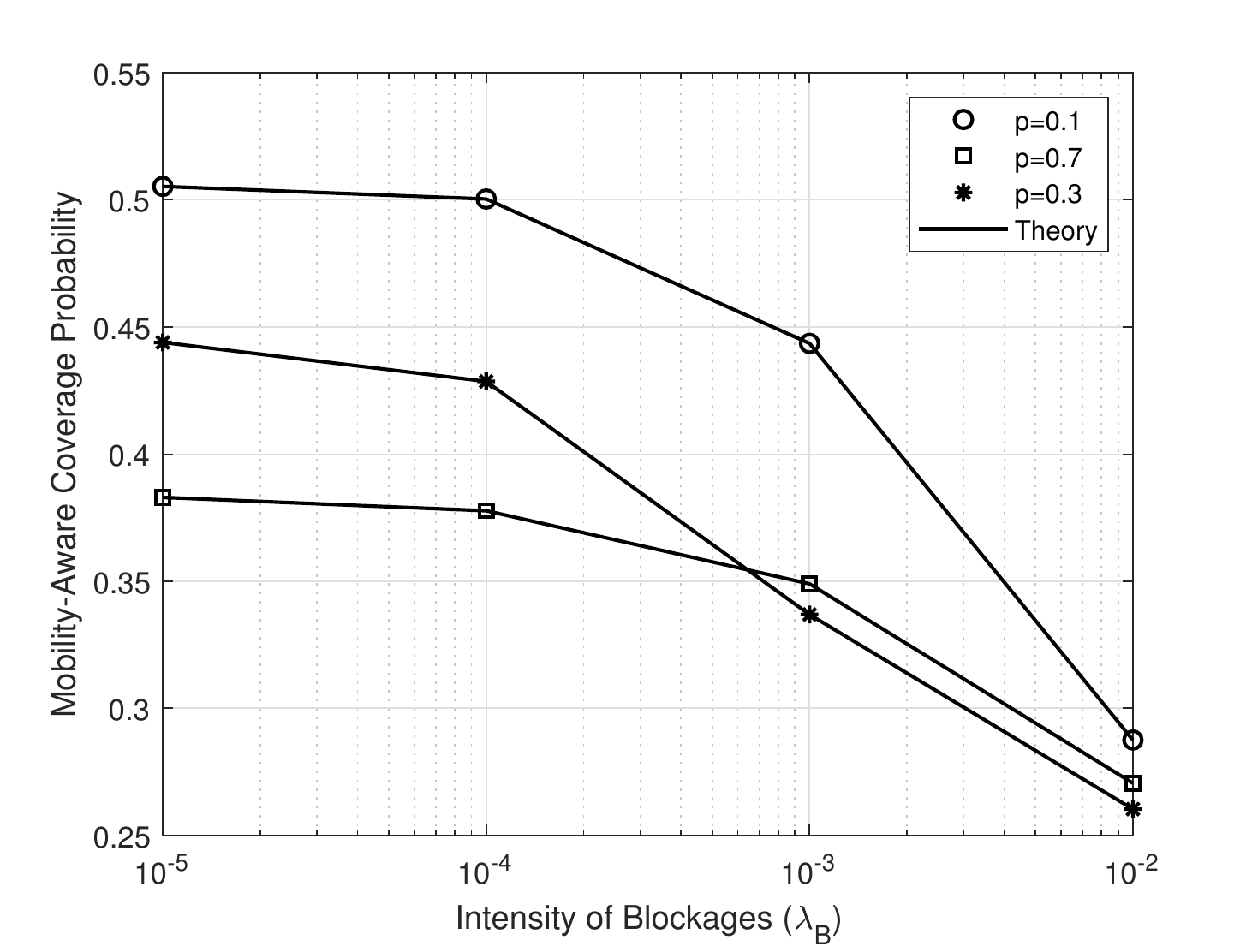}\label{fig:f9}}
  \caption{{Overall coverage probability of a typical user with mobility  as a function of the intensity of blockages.}}
  \label{blockage}
\end{figure}

\section{Conclusions and Future Work}
In this article, we provided a comprehensive stochastic geometry framework to describe the overall performance of a mobile user in a two-tier hybrid RF-THz network. We derived novel coverage probability expressions considering molecular noise in THz transmissions and derived the coverage probability with mobility. We validated the accuracy of derived expressions using Monte-Carlo simulations. 
Our numerical  results depict that the probability of HO in THz network is of much more significance than in conventional RF network, especially for lower molecular absorption coefficients. Therefore, mobility-aware performance frameworks are of immediate relevance. Also, our results demonstrated that the benefits of interference reduction due to increasing molecular absorption can dominate compared to the signal degradation at reasonable intensity of TBSs. This is favourable news for upcoming 6G networks. Furthermore, our results revealed that  ignoring molecular absorption and mobility can lead to significantly over-optimistic results, especially in high frequency  THz networks. 
\textcolor{black}{
In this paper, we considered the standard procedure of performing handoff based on the received signal power measurements from different access points. However, developing sophisticated handoff mechanisms for short-range transmissions is an interesting research topic for further investigation.}
Furthermore, to capture the distinct coverage zones, transmit powers,  deployment intensity, and channel propagation, we consider complementing THz with RF. However, the proposed framework can be extended for mm-wave based model by modifying the path-loss model of RF or THz, as needed.

\begin{appendices}
\section{Proof of Lemma~1}
\renewcommand{\theequation}{A.\arabic{equation}}
\setcounter{equation}{0}
Given the event $k=T$ is defined as the event when user is associated to TBS, then  the conditional PDF of the distance from the connected TBS can be {obtained} as follows:
\begin{equation}
\label{eq:Distance_distribution_Proof}
f_{r_{T}}\left(r_T \right) = \frac{1}{A_T}  \frac{d\mathbb{P}(d_0 > r_T, k = T)}{dd_0}.
\end{equation}
Subsequently, the joint PDF in the numerator can be derived as follows:
\begin{equation}
\label{eq:Distance_distribution_Proof2}
\begin{split}
&\mathbb{P}(d_0> r_T,  k  = {\rm T})  = \mathbb{P}( d_0 > r_T,  P_{T}^{\mathrm{rx}} > P_{R}^{\mathrm{rx}}), \\
& =   \int_{r_T}^{\infty}\mathbb{P}(P_{T}^{\mathrm{rx}} > P_{R}^{\mathrm{rx}}) f_{d_0}(d_0) d d_0, \\
& =   \int_{r_T}^{\infty}\mathbb{P}\left(P_{{T}}\gamma_{T}\frac{\exp\left(-K_{a} d_0\right)}{d_0^2}>P_{{R}}\gamma_{R} r_0^{-\alpha}\right) f_{d_0}(d_0)d d_0, \\
& \stackrel{(a)}{=} \int_{r_T}^{\infty} 2\pi \lambda_{T} d_0\:e^{\left(-\pi \lambda_T d_0^2 - \pi \lambda_R (d_0^2 Q)^{\frac{2}{\alpha}} \exp\left(\frac{2 K_a d_0}{\alpha}\right) \right)}d d_0,
\end{split}
\end{equation}
\normalsize
where (a) follows from substituting $\mathbb{P}(P_{T}^{\mathrm{rx}} > P_{R}^{\mathrm{rx}})$  in \eqref{eq:Asso_Probab_Proof}, and $f_{d_0}(d_0) = 2\pi \lambda_{T} d_0\exp\left(-\pi\lambda_{T} d_0^2\right)$. Now {the final value of $f_{r_{T}}\left(r_T\right)$ is obtained by replacing \eqref{eq:Distance_distribution_Proof2} into \eqref{eq:Distance_distribution_Proof}. }
{Similarly}, the event $k=R$ is defined as the event when user is associated to RBS, then  the conditional PDF of the distance from the tagged RBS can be derived as follows:
\begin{equation}
\label{eq:Distance_distribution_Proof0}
\begin{split}
f_{r_{R}}\left(r_R \right) = \frac{1}{A_R}  \frac{d\mathbb{P}(r_0 > r_R, k = R)}{dr_R}.
\end{split}
\end{equation}
Subsequently, the joint PDF in the numerator of \eqref{eq:Distance_distribution_Proof0} can be derived as follows:
\begin{equation}
\label{eq:Distance_distribution_Proof3}
\begin{split}
&  \mathbb{P}( r_0 > r_R, k = {\rm R}) = \mathbb{P}( r_0 > r_R, P_{R}^{\mathrm{rx}} > P_{T}^{\mathrm{rx}}),\\
 & = \int_{r_R}^{\infty}\mathbb{P}(P_{R}^{\mathrm{rx}} > P_{T}^{\mathrm{rx}}) f_{r_0}(r_0) dr_0, \\
& = \int_{r_R}^{\infty}\mathbb{P}\left( P_{{R}}\gamma_{R} r_0^{-\alpha} > P_{{T}}\gamma_{T}\frac{\exp\left(-K_{a} d_0\right)}{d_0^2} \right) f_{r_0}(r_0) dr_0, \\
& \stackrel{(a)}{\approx} \int_{r_R}^{\infty} \mathbb{P}\left( P_{{R}}\gamma_{R} r_0^{-\alpha}> P_{{T}}\gamma_{T}{d_0^{-2-\mu}} \right) f_{r_0}(r_0) dr_0, \\
& = \int_{r_R}^{\infty} 2\pi \lambda_{R} r_0 \exp\left(- \pi \lambda_{R} r_0^2 -\pi \lambda_{T} \left(\frac{ r_0^{\alpha}}{Q}\right)^{\frac{2}{2+\mu}} \right) dr_0,
\end{split}
\end{equation}
{where the approximation $r_T^2\exp\left(K_{a} r_T\right)$ with $r^{2 + \mu}$ and the efficient choice of correcting factor $\mu$ (as discussed in Section~III) help to derive the expression in step (a).} Finally, substituting \eqref{eq:Distance_distribution_Proof3} in \eqref{eq:Distance_distribution_Proof} yields $f_{r_{T}}\left(r_T\right)$ as given in {\bf Lemma~1}. Here, $f_{r_0}(r_0) = 2\pi \lambda_{R} r_0\exp\left(-\pi\lambda_{R} r_0^2\right)$ and $f_{d_0}(d_0) = 2\pi \lambda_{R} d_0\exp\left(-\pi\lambda_{R} d_0^2\right)$ .

\begin{figure*}
\centering
\includegraphics[width=6.5in]{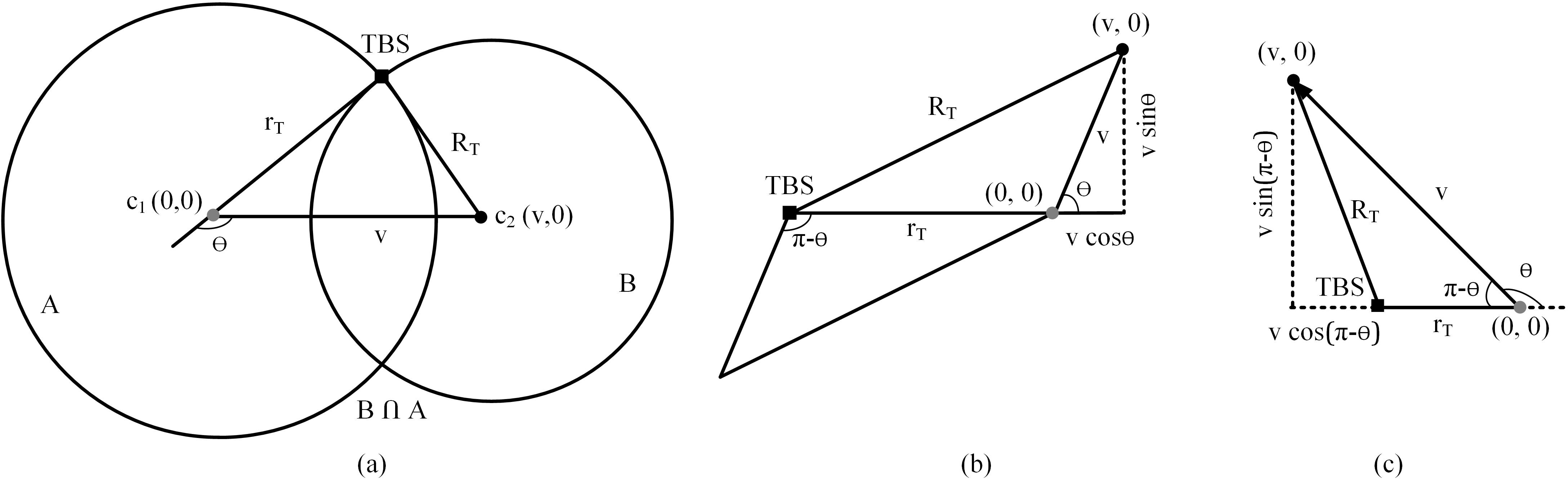}
\caption{{(a) Geometrical illustration of $r_T, R_T, v $ and $\theta$, (b) $0\leq \theta < \frac{\pi}{2}$, (c) $\mathrm{\frac{\pi}{2} \leq \theta \leq \pi}$}.}
\label{sys_mod_exp}
\end{figure*}

\section{Proof of Lemma~2}
\renewcommand{\theequation}{B.\arabic{equation}}
\setcounter{equation}{0}
It can be observed from Fig.~2(a) that the vertical HO between RF and THz tiers does not {occur} if all BSs within these two tiers except {tagged} BS (i.e., TBS) are {situated} outside the area $|B^{\prime} \setminus B^{\prime}\cap A^{\prime}|$. {When the} distance is $r_T$ and {the direction (or angle) of user movement is} $\theta$, there will be no HO from the serving TBS with probability as follows:
\begin{align} \label{t2r_nhop}
    \mathbb{P}(\overline{H}_{T}|r_{T},\theta)  =  \mathbb{P}(N(|B^{\prime} \setminus B^{\prime} \cap A^{\prime}|) = 0|T_{T} \neq T_{R})
    \nonumber\\ + \mathbb{P}(N(|B\setminus B\cap A|) = 0|T_{T} = T_{R}),
\end{align}
where $\overline{H}_{T}$ is the \textcolor{black}{complement} of $H_{T}$ and $N(\cdot)$ \textcolor{black}{presents} the number of BSs \textcolor{black}{within a particular} area. Here, $T_T$ and $T_R$ simplifies the THz and RF tier, respectively. After averaging over $r_T$ and $\theta$, $\mathbb{P}[H_{T}]$  is given in \textbf{Lemma~2}.
\textcolor{black}{The first expression in } Eq. (\ref{t2r_nhop}) \textcolor{black}{states} the vertical HO and the second \textcolor{black}{expression} \textcolor{black}{is} the horizontal HO. \textcolor{black}{Finally,} applying the null \textcolor{black}{property} of the PPP, we have:
\begin{align} \label{t2r_nhop_null}
    \mathbb{P}(\overline{H}_{T}|r_{T},\theta) & 
     = e^{\lambda_{T}  |B\setminus B\cap A| + \lambda_{R}  |B^{\prime} \setminus B^{\prime} \cap A^{\prime}|}.
\end{align}
Here, $B = \pi R_{T}^{2}$, $B^{\prime} = \pi R_{R}^{\prime 2}$. {From Fig.~\ref{sys_mod_exp}, the part of expression in Eq. (\ref{t2r_nhop_null}), $|B\cap A|$ can be easily determined as follows:}
\small
\begin{align}\label{BA}
    &|B\cap A|  =  R_{T}^{2} \mathrm{cos}^{-1} \left( \frac{R_{T}^{2}+v^2-r_{T}^{2}}{2R_{T}^{2}v} \right) + r_{T}^{2} \mathrm{cos}^{-1} \left( \frac{r_{T}^{2}+v^2-R_{T}^{2}}{2r_{T}^{2}v} \right) \nonumber\\
    & - \frac{1}{2} \sqrt{(r_{T}+R_{T}-v)(r_{T}+R_{T}+v)(v+r_{T}-R_{T})(v-r_{T}+R_{T})},\nonumber\\
    & = R_{T}^{2} \:\theta_{1}^{T} + r_{T}^{2}(\pi - \theta) - r_{T}\:v \:\mathrm{sin}{\theta}.
\end{align}
\normalsize
{Similarly, the term $|B^{\prime} \cap A^{\prime}|$ can also be calculated from Fig.~\ref{sys_mod_exp} as follows:}
\small
\begin{align} \label{BA_prime}
    &|B^{\prime} \cap A^{\prime}|  =  R_{R}^{\prime 2} \mathrm{cos}^{-1} \left( \frac{R_{R}^{\prime 2}+v^2-r_{R}^{\prime 2}}{2R_{R}^{\prime 2}v} \right) + r_{R}^{\prime 2} \theta_{2}^{T} \nonumber\\
    & - \frac{1}{2} \sqrt{(r_{R}^{\prime}+R_{R}^{\prime}-v)(r_{R}^{\prime}+R_{R}^{\prime}+v)(v+r_{R}^{\prime}-R_{R}^{\prime})(v-r_{R}^{\prime}+R_{R}^{\prime})},\nonumber\\
    & = R_{R}^{\prime 2} \theta_{3}^{T} + r_{R}^{\prime 2}(\pi - \theta_{2}^{T}) - r_{R}^{\prime} \:v\: \mathrm{sin}{\theta_{2}^{T}}.
\end{align}
\normalsize
The common area between two intersecting tiers can be calculated from Eq. (\ref{BA}) and Eq. (\ref{BA_prime}), where 
$ R_{T}^{2} = r_{T}^{2} + v^{2}-2r_{T}v\mathrm{cos}(\pi - \theta),
 R_{R}^{\prime} = \left(R_{T}\right)^{\frac{2}{\alpha}}\:e^{ \frac{K_{a}\:R_{T}}{\alpha}} \left( \frac{P_{R}^{tx} Q}{P_{T}^{tx}} \right)^{\frac{1}{\alpha}} , r_{R}^{\prime} = \left(r_{T}\right)^{\frac{2}{\alpha}}\:e^{ \frac{K_{a}\:R_{T}}{\alpha}}\left( \frac{P_{R}^{tx} Q}{P_{T}^{tx}} \right)^{\frac{1}{\alpha}}, \theta_{1}^{T} = \theta - \mathrm{sin}^{-1} \left(\frac{v\:\mathrm{sin}\theta}{R_T}\right), \theta_{2}^{T} = \mathrm{cos}^{-1} \left( \frac{r_{R}^{\prime 2} + v^2 - R_{R}^{\prime 2}}{2r_{R}^{\prime} v} \right), \theta_{3}^{T} = \theta - \mathrm{sin}^{-1} \left(\frac{v\:\mathrm{sin}\theta}{R_{R}^{\prime}} \right) $.


According to the Fig.~2(a), the common area between two intersecting circles of radii $r_T$ and $R_T$ is $S_{T}$. 
Here, the value of $\theta_1^T = \theta - \mathrm{sin}^{-1} \left( \frac{v\: \mathrm{sin}\theta}{R_T} \right)$, which is true when the value of $ \theta $ lies between 0 and $ \frac{\pi}{2} $. According to Fig.~\ref{sys_mod_exp}(c), when the value of $ \theta $ is in between $ \frac{\pi}{2} $ and $\pi$ then $r$ is no longer greater than $\left\{v\: \mathrm{ cos }(\pi - \theta)\right\}$ or in other words, $v\: \mathrm{ cos } (\pi - \theta) > r$. For $ \frac{\pi}{2} \leq \theta \leq \pi $ from \cite{tabassum2019fundamentals}, $\left\{\mathrm{sin}^{-1} \left( \frac{v\: \mathrm{sin}\theta}{R_T} \right )\right\}$ in $\theta_1^T$ will be replaced by $\left\{\pi - \mathrm{sin}^{-1} \left( \frac{v\: \mathrm{sin}\theta}{R_T} \right )\right\}$. Therefore, we define a new term  $C_T$ with modified $\theta_1^T$ to substitute the term  $S_T$ in {\bf Lemma~2}. 

 On the contrary, the common area between two intersecting circles from radii $r_{R}^{\prime}$ and $R_{R}^{\prime}$ is $S_{T}^{\prime}$. 
 For $ \frac{\pi}{2} \leq \theta_2^T \leq \pi $ from \cite{tabassum2019fundamentals}, $\left\{\mathrm{sin}^{-1} \left( \frac{v\: \mathrm{sin}\theta_2^T}{R^{\prime }_{R}} \right )\right\}$ in $\theta_3^T$ will be replaced by $\left\{\pi - \mathrm{sin}^{-1} \left( \frac{v\: \mathrm{sin}\theta_2^T}{R^{\prime}_{ R}} \right )\right\}$. Therefore, we define a new term  $C^{\prime}_{ T}$ with modified $\theta_3^T$ to substitute the term  $S^{\prime}_{T}$ in {\bf Lemma~2}.

\section{Proof of Lemma~4}
\renewcommand{\theequation}{C.\arabic{equation}}
\setcounter{equation}{0}
The vertical HO between THz and RF tiers does not \textcolor{black}{occur} if all BSs  except \textcolor{black}{tagged} RBSs are \textcolor{black}{located} outside the area $|A^{\prime} \setminus A^{\prime} \cap B^{\prime}|$. \textcolor{black}{When the} distance is $r_R$ and the direction \textcolor{black}{(or angle) of the user movement is} $\theta$, there \textcolor{black}{will be} no HO \textcolor{black}{from the serving RBS} with probability as follows:
\begin{align} \label{r2t_nhop}
    &\mathbb{P}(\overline{H}_{R}|r_{R},\theta)  =  \mathbb{P}(N(|A^{\prime} \setminus A^{\prime} \cap B^{\prime} |) = 0|T_{R} \neq T_{T})
     \\&+ \mathbb{P}(N(|A\setminus A\cap B|) = 0|T_{R} = T_{T}),
\end{align}
where $\overline{H}_{R}$ is the \textcolor{black}{complement} of $H_{R}$ and $N(\cdot)$ \textcolor{black}{depicts the number of BSs  within a particular} area. \textcolor{black}{The first expression in} Eq. (\ref{r2t_nhop}) \textcolor{black}{states} the vertical HO and the second \textcolor{black}{expression determines} the horizontal HO. \textcolor{black}{Finally,} applying the null \textcolor{black}{property} of the PPP, we have:
\begin{align} \label{r2t_nhop_null}
    \mathbb{P}(\overline{H}_{R}|r_{R},\theta) = \mathrm{exp}\left(\lambda_{R} \cdot |A\setminus A\cap B| + \lambda_{T} \cdot |A^{\prime} \setminus A^{\prime} \cap B^{\prime} |\right).
\end{align}
Here, $A = \pi R_{R}^{2}$, $A^{\prime} = \pi R_{T}^{\prime 2}$. {Likewise appendix B, it can be determined the two-parts of Eq. (\ref{r2t_nhop_null}), i.e., $|A\cap B|$ and $|A^{\prime} \cap B^{\prime}|$ can be given as in the following:}
\small
\begin{align}\label{AB}
    &|A\cap B|  = R_{R}^{2} \mathrm{cos}^{-1} \left( \frac{R_{R}^{2}+v^2-r_{R}^{2}}{2R_{R}^{2}v} \right) + r_{R}^{2} \mathrm{cos}^{-1} \left( \frac{r_{R}^{2}+v^2-R_{R}^{2}}{2r_{R}^{2}v} \right) \nonumber\\
    & - \frac{1}{2} \sqrt{(r_{R}+R_{R}-v)(r_{R}+R_{R}+v)(v+r_{R}-R_{R}) (v-r_{R}+R_{R})}, \nonumber\\
    & = R_{R}^{2} \:\theta_{1}^{R} + r_{R}^{2}(\pi - \theta) - r_{R}\:v \:\mathrm{sin}{\theta}.
\end{align}
\begin{align}  \label{AB_prime}
    &|A^{\prime} \cap B^{\prime}|  =  R_{T}^{\prime 2} \mathrm{cos}^{-1} \left( \frac{R_{T}^{\prime 2}+v^2-r_{T}^{\prime 2}}{2R_{T}^{\prime 2}v} \right) + r_{T}^{\prime 2} \theta_{2}^{R} \nonumber\\
    & - \frac{1}{2} \sqrt{(r_{T}^{\prime}+R_{T}^{\prime}-v)(r_{T}^{\prime}+R_{T}^{\prime}+v)(v+r_{T}^{\prime}-R_{T}^{\prime}) (v-r_{T}^{\prime}+R_{T}^{\prime})}, \nonumber\\
    & = R_{T}^{\prime 2} \:\theta_{3}^{R} + r_{T}^{\prime 2}(\pi - \theta_{2}^{R}) - r_{T}^{\prime}\:v \:\mathrm{sin}{\theta_{2}^{R}}.
\end{align}
\normalsize



The common area between two intersecting tiers can be calculated from Eq. (\ref{AB}) and Eq. (\ref{AB_prime}), where $R_{R}^{2} = r_{R}^{2} + v^{2}-2r_{R} \: v \: \mathrm{cos}(\pi - \theta), R_{T}^{\prime}=\left[ (R_{R})^{\alpha}\:(\frac{P_{T}^{tx}}{P_{R}^{tx}\:Q}) \right]^{\frac{1}{2+\mu}}, r_{T}^{\prime}=\left[ (r_{R})^{\alpha}\:\left(\frac{P_{T}^{tx}}{P_{R}^{tx}\:Q}\right)\right]^{\frac{1}{2+\mu}}$, 
$\theta_{1}^{R} = \theta - \mathrm{sin}^{-1} \left(\frac{v\:\mathrm{sin}\theta}{R_R} \right), \theta_{2}^{R} = \mathrm{cos}^{-1} \left( \frac{r_{T}^{\prime 2} + v^2 - R_{T}^{\prime 2}}{2r_{T}^{\prime} v} \right), \theta_{3}^{R} = \theta - \mathrm{sin}^{-1} \left(\frac{v\:\mathrm{sin}\theta}{R_{T}^{\prime}} \right) $.



According to the Fig. \ref{u3}(b), the common area between two intersecting circle with radii $r_R$ and $R_R$ is $S_{R}$. 
Here, the value of $\theta_1^R = \theta - \mathrm{sin}^{-1} \left( \frac{v\: \mathrm{sin}\theta}{R_R} \right)$, which is true when the value of $ \theta $ lies between 0 and $ \frac{\pi}{2} $. Similar way as Fig.~\ref{sys_mod_exp} (c), when $ \frac{\pi}{2} \leq \theta \leq \pi $ then, $v\: \mathrm{ cos } (\pi - \theta) > r$. For $ \frac{\pi}{2} \leq \theta \leq \pi $ from \cite{tabassum2019fundamentals}, $\theta_1^R $ becomes $\left\{\theta - \pi + \mathrm{sin}^{-1} \left( \frac{v\: \mathrm{sin}\theta}{R_R} \right )\right\}$. {Therefore, a new term $C_R$ is defined by modifying $\theta_1^R$ to substitute the term  $S_R$ in {\bf Lemma~4}.}

Likewise, the common area between two intersecting circles of radii $r_{T}^{\prime}$ and $R_{T}^{\prime}$ is $S_{R}^{\prime}$,
{and later term  $C_R^{\prime}$ with modified $\theta_3^R$ to substitute the term  $S_R^{\prime}$ in {\bf Lemma~4}.}

\end{appendices}

\bibliographystyle{IEEEtran}
\bibliography{IEEEabrv,reference}
\begin{IEEEbiography}[{\includegraphics[width=1in,height
=1.25in,clip,keepaspectratio]{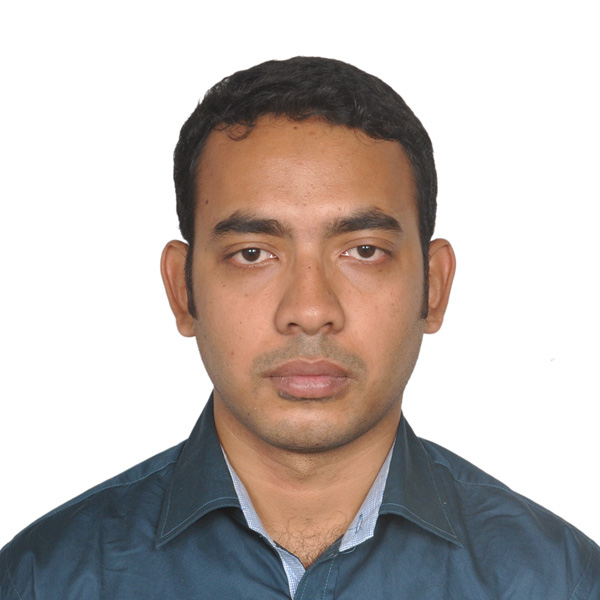}}]{Md. Tanvir Hossan}
Md Tanvir Hossan received his MASc degree in Electrical and Computer Engineering from York
University, Canada in April 2021. He also received his MSc degree in Electronics Engineering from
Kookmin University, South Korea, in 2018. His research interests include wireless communications
including space communications, vehicular communications, and machine learning.
\end{IEEEbiography}

\begin{IEEEbiography}[{\includegraphics[width=1.1in,height=1.35in]{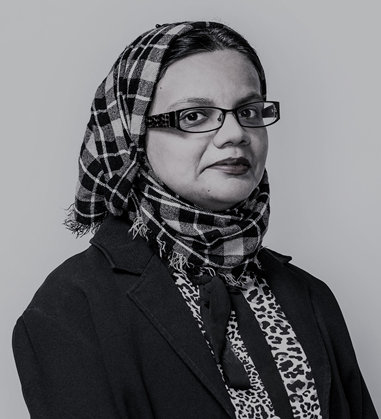}}]
{Hina Tabassum} (SM'17) is currently an Assistant Professor at the Lassonde School of Engineering, York University, Canada. Prior to that, she was a postdoctoral research associate at the Department of Electrical and Computer Engineering, University of Manitoba, Canada. She received her PhD degree from King Abdullah University of Science and Technology (KAUST). She is a Senior member of IEEE and registered Professional Engineer in the province of Ontario, Canada. She is the founding chair of a special interest group on THz communications in IEEE Communications Society (ComSoc) - Radio Communications Committee (RCC). She has been recognized as an Exemplary Reviewer (Top 2\% of all reviewers) by IEEE Transactions on Communications in 2015, 2016, 2017, 2019, and 2020. She has been recognized as an Exemplary Editor by IEEE Communications Letters, 2020. Currently, she is serving as an Associate Editor in IEEE Communications Letters, IEEE Transactions on Green Communications, IEEE Communications Surveys and Tutorials, and IEEE Open Journal of Communications Society.  Her research interests include stochastic modeling and optimization of wireless networks including vehicular, aerial, and satellite networks, millimeter and terahertz communication networks, machine learning empowered resource allocation in wireless networks.

\end{IEEEbiography}

\end{document}